\documentclass[AMA,STIX1COL]{WileyNJD-v2}
\usepackage{moreverb}
\usepackage{amsfonts}
\usepackage{mathrsfs}
\usepackage{amsmath}
\usepackage{bm}
\usepackage{graphicx}
\usepackage{subfigure}
\usepackage{color}
\usepackage{verbatim}
\usepackage{pgf,tikz,epstopdf}
\usetikzlibrary{arrows,decorations.pathmorphing, backgrounds ,positioning, fit, matrix}
\usepackage[ruled,lined]{algorithm2e}
\usepackage{algpseudocode}

\usepackage{natbib}
\hypersetup{
colorlinks=true,
linkcolor=blue,
citecolor=blue
}

\newcommand{\wxy}[1]{\textcolor{black}{#1}}
\newcommand{\hw}[1]{\textcolor{black}{#1}}

\newtheorem{thm}{Theorem}
\newcommand{\col}{\textnormal{col}}
\newcommand{\diag}{\textnormal{diag}}

\newcommand\BibTeX{{\rmfamily B\kern-.05em \textsc{i\kern-.025em b}\kern-.08em
T\kern-.1667em\lower.7ex\hbox{E}\kern-.125emX}}

\articletype{research article}%

\received{<day> <Month>, <year>}
\revised{<day> <Month>, <year>}
\accepted{<day> <Month>, <year>}

\begin{document}

\title{\wxy{Strategic learning for disturbance rejection in multi-agent systems: Nash and Minmax in graphical games}}

\author[1]{Xinyang Wang}

\author[2]{Martin Guay}

\author[3]{Shimin Wang}

\author[1]{Hongwei Zhang}

\authormark{Wang \textsc{et al}}

\address[1]{\orgdiv{Guangdong Provincial Key Laboratory of Intelligent Morphing Mechanisms and Adaptive Robotics}, \orgdiv{School of Mechanical Engineering and Automation}, \orgname{Harbin Institute of Technology}, \orgaddress{\state{Shenzhen, Guangdong 518055}, \country{P.R. China}}}

\address[2]{\orgdiv{Department of Chemical Engineering}, \orgname{Queen's University}, \orgaddress{\state{Kingston, ON K7L 3N6}, \country{Canada}}}

\address[3]{\orgname{Massachusetts Institute of Technology}, \orgaddress{\state{Cambridge, MA 02139}, \country{USA}}}

\corres{*Hongwei Zhang, Harbin
Institute of Technology, Shenzhen,
Guangdong 518055, P.R. China.\\ \email{hwzhang@hit.edu.cn}}

\fundingInfo{the National Key R{\&}D Program of China under Grant 2022YFB4700200, Guangdong Basic and Applied Basic Research Foundation under project 2023A1515011981, the Shenzhen Science
and Technology Program under projects 
JCYJ20220818102416036 and RCJC20210609104400005, and partly by NSERC.}


\abstract[Abstract]{\wxy{This article investigates the optimal control problem with disturbance rejection for discrete-time multi-agent systems under cooperative and non-cooperative graphical games frameworks.
Given the practical challenges of obtaining accurate models, Q-function-based policy iteration methods are proposed to seek the Nash equilibrium solution for the cooperative graphical game and the distributed minmax solution for the non-cooperative graphical game.
To implement these methods online, two reinforcement learning frameworks are developed, an actor-disturber-critic structure for the cooperative graphical game and an actor-adversary-disturber-critic structure for the non-cooperative graphical game.
The stability of the proposed methods is rigorously analyzed, and simulation results are provided to illustrate the effectiveness of the proposed methods.}}


\jnlcitation{\cname{
\author{Wang X}, 
\author{Guay M},
\author{Wang S},
\author{Zhang H}.}
\ctitle{Strategic learning for disturbance rejection in multi-agent systems: Nash and Minmax in graphical games}. \cjournal{Int J Robust Nonlinear Control.}\cvol{}}

\maketitle

\section{Introduction}\label{sec1}
Over the past two decades, distributed control of multi-agent systems (MAS) has attracted extensive attention from both academia and industry. 
It has been applied in various scenarios, such as smart grid \cite{yaz2014}, mobile robots \cite{wilson2020robotarium} and unmanned air vehicles \cite{wang2018finite}. 
So far, fruitful results have been reported in the control community \cite{jadbabaie2003coordination,olfati2004consensus,ren2005consensus,hong2006tracking}.
The objective of distributed control is to design a control protocol for each agent via local interactions with neighbours, reaching a common control goal of the overall system. 
Compared with centralized control scheme, distributed control has greater flexibility, scalability and robustness \citep*{olfati2007consensus}. 
By distributing the control objects, MAS can enhance fault tolerance against some unexpected emergencies,  lower communication costs and reduce the complexity of designing a control protocol.  
Optimal distributed control systems are designed to minimize a cost function in a distributed fashion. 
The formal construction of the control policy of an optimal distributed control problem requires the solution of the Hamilton-Jacobi-Bellman (HJB) equation, which is generally impossible to compute analytically. 
%

%
To numerically solve the HJB equation, reinforcement learning (RL) is usually adopted. 
%
RL is a computational technique, whose principle is to optimize a control policy by maximizing the expected cumulative reward from the current state to the final state based on a trial-and-error approach \citep*{sutton2018reinforcement,tang2022robust}. 
Model-based RL methods, such as policy iteration (PI) and value iteration (VI), require complete knowledge of the system's dynamics, which limits the applicability of RL techniques \citep*{bertsekas1996neuro}.
To address this issue, a Q-learning algorithm has been proposed to compute optimal control policies in the absence of any specific knowledge of the system dynamics\citep*{al2007model}.
In control engineering, the RL technique is primarily utilized to solve the HJB equation to minimize a cost function. 
They have been exploited to overcome difficulties in many practical situations. Of particular interest in this study are RL methods that overcome the difficulties associated with the optimal distributed control problem \citep*{peng2022distributed}.
The combination of RL with adaptive control techniques has led to the development of several methods to learn optimal control policies in real-time, including adaptive dynamic programming (ADP), heuristic dynamic programming (HDP) and others \citep*{lewis2009reinforcement}.
So far, many control methods based on RL, such as event-triggered control, $H_{\infty}$ control, among others, have been studied with some progress \citep*{kiumarsi2017optimal}.

In practice, most physical systems are inevitably affected by uncertain disturbances.
$H_{\infty}$ control is a well-established robust control method to minimize the effects of uncertainties and disturbances for dynamic systems \citep*{isidori1994h}. 
In $H_{\infty}$ control, the $H_{\infty}$ norm of the transfer function from disturbances to output, which can quantify the worst-case disturbance attenuation, is minimized to improve the stability and performance of control systems.
By considering the disturbance input as a maximizing player and the control input as a minimizing player, $H_{\infty}$ control can be transformed into a zero-sum game which can be solved using RL techniques \citep*{bacsar2008h}. 
A model-free Q-learning method is developed to deal with the $H_{\infty}$ tracking problem by solving the game algebraic Riccati equation in Reference \zcite{yang2021h}.
The optimal $H_{\infty}$ control problem is solved by an output feedback Q-learning algorithm for the linear zero-sum game using an off-line method in Reference \zcite{rizvi2018output}.
The above-mentioned works on RL based $H_{\infty}$ control mainly focus on single-agent systems.

\hw{For complex and challenging tasks, MAS generally outperform single-agent systems.
It is well established that by combining RL and graphical game theory, optimal control problems for MAS can be solved effectively \cite{vamvoudakis2010online,abouheaf2014multi,zhang2016data}.
According to the relationship between each agent in MAS, graphical games can be cooperative or non-cooperative.
For the case where the knowledge of the neighbors' actions is available, each agent of a cooperative graphical games can implement the best response, and achieve Nash equilibrium\cite{abouheaf2014multi}.
In the case where the neighbors' actions are unknown, each agent in a non-cooperative graphical games must make the best decision subject to the worst-case actions of its neighbors to achieve conservative performance in a fully distributed fashion\cite{lopez2020reinforcement}.
%
%
Recently, graphical game theory and RL techniques have also been utilized to address distributed $H_{\infty}$ control problem of MAS.
$H_{\infty}$ control of MAS is considered by solving the coupled HJI equation for the differential cooperative graphical game in Reference \zcite{jiao2016multi}, which requires knowledge of the overall system dynamics. 
For homogeneous continuous-time MAS where neighbors' actions are unavailable, a non-cooperative graphical game with disturbance rejection is formulated and implemented by seeking a distributed minmax solution\cite{lian2022online}.
These methods to acquire the solution of cooperative and non-cooperative graphical games all need knowledge of the system dynamics.
However, obtaining accurate mathematical models for physical systems, if possible, is usually challenging.
To date, both cooperative and non-cooperative graphical games for discrete-time MAS with disturbance rejection in the absence of accurate system models has not been addressed satisfactorily in the literature.}
\wxy{
In this article, both cooperative and non-cooperative graphical games for discrete-time MAS are addressed using PI algorithms that achieve a Nash equilibrium solution in the cooperative case and a fully distributed minmax solution with disturbance rejection in the non-cooperative case.
The proposed PI methods are based on the Q-function which requires no information about the system dynamics.
To seek the Nash equilibrium solution and the fully distributed minmax solution online, an actor-disturber-critic framework and an actor-adversary-disturber-critic framework are proposed for cooperative and non-cooperative graphical games, respectively.
The convergence to the approximate Nash solution and the approximate distributed minmax solution of the online learning algorithms are rigorously analyzed.}


%
\emph{\textbf{Notations:}} The $n \times n$ identity matrix is denoted by $I_n$.
The notation $\mathbf{1}_n$ is the n dimensional vector of ones.
The maximum and minimum singular values of a matrix are denoted by $\bar{\sigma}(\cdot)$ and $\underline{\sigma}(\cdot)$, respectively.
The vector of columns of a matrix is denoted by $\textnormal{vec}(\cdot)$.
Matrix $\diag(x_1,\cdots,x_n)$ is diagonal with $x_i$ being the $i^{th}$ diagonal entry.
Matrix $P>0$ ($P \geq 0$) means $P$ is positive definite (semi-definite).
Kronecker product is denoted by $\otimes$.
Both Euclidean norm of a vector and Frobenius norm of a matrix are denoted by $\lVert \cdot \rVert$.
\section{Preliminaries and problem formulation} \label{Sec Preliminaries}

\subsection{Graph theory} \label{subsec-graph}
A directed graph $\mathcal G$ consists of a finite nonempty set of nodes $\mathcal V = \{v_1, \dots, v_N\}$ and a set of edges $E \subseteq \mathcal V \times \mathcal V $. 
Its adjacent matrix is denoted as $\mathcal A =[a_{ij}]$ where $a_{ij} > 0$ if $(v_i, v_j) \in E$ and $a_{ij} = 0$ otherwise.
Node $v_j$ is a neighbor of $v_i$ if $a_{ij} > 0$. 
The set of neighbors of node $i$ is denoted as $N_i = \{j|a_{ij} > 0\}$ and $-i=\{j|j \in N_i\}$.
The Laplacian matrix is defined as $\mathcal L = \mathcal D - \mathcal A$, where $\mathcal D = \diag(d_1,\cdots,d_N)$ is the in-degree matrix with $d_i = \sum _ {j \in N_i}a_{ij}$.
A sequence of edges $\{ (v_i, v_k), (v_k, v_l), \dots, (v_l, v_j) \}$ is a directed path from node $i$ to node $j$. 
A directed graph contains a spanning tree if the root node has a directed path to every other node.

\subsection{Problem formulation} \label{Sec-PF}
Consider a MAS consisting of a leader node and $N$ follower nodes.
The follower node is described by
{
\begin{align} \label{e2.2.1}
x_{i,k+1} = Ax_{i,k} + B_iu_{i,k}+ E_iw_{i,k},
\end{align}
where $x_{i,k}\in \mathbb{R}^n$ is the state vector, $u_{i,k} \in \mathbb{R}^p$ is the control input, $w_{i,k} \in \mathbb{R}^q$ is the disturbance input vector belonging to $L_2[0,\infty)$, i.e., $\sum_{k=0}^{\infty} \lVert w_{i,k} \rVert^2 < \infty$.}
The system matrices $A \in \mathbb{R}^{n \times n}$, $B_i \in \mathbb{R}^{n \times p}$, $E_i \in \mathbb{R}^{n \times q}$ are constants and $(A,B_i)$ is reachable for $i=1,2,\dots,N$.
The leader agent is
\begin{align} \label{e2.2.2}
x_{0,k+1} = Ax_{0,k}.
\end{align}
The local neighborhood tracking error $\delta_i\in \mathbb{R}^n$ of agent $i$ is defined as
\begin{align} \label{e2.2.3}
\delta_{i,k}  =\sum\nolimits_{j \in N_i}a_{ij}(x_{j,k}-x_{i,k})+g_i(x_{0,k}-x_{i,k})
\end{align}
where $g_i = 1$ if agent $i$ is pinned to the leader, otherwise $g_i = 0$. The compact form of local neighbourhood tracking error is
\begin{align*} 
\delta_k=-((L+G)\otimes I_n)\epsilon_k
\end{align*}
where $\epsilon_k=x_k-\underline{x}_{0,k}$ is the global disagreement vector,
$x_k=\col(x_{1,k},x_{2,k},\dots,x_{N,k}) \in \mathbb{R}^{Nn}$,
$\underline{x}_{0,k} = \mathbf{1}_n \otimes x_{0,k}$
and $G = \textnormal{diag}(g_1,\dots,g_N )$ is the pinning gain matrix.

\begin{assumption} \label{a1}
 The graph contains a spanning tree and $g_i = 1$ if agent $i$ is the root of the spanning tree.
\end{assumption} 


\begin{lemma}\label{l1}
\citep*{zhang2012adaptive}Under Assumption \ref{a1}, the global disagreement vector $ \epsilon_k$ is bounded by $\lVert \delta_k\rVert /\underline{\sigma}((L+G)\otimes I_n)$.
\end{lemma}
The dynamics of agent $i$'s local neighbourhood tracking error is
\begin{align} \label{e2.2.8}
\delta_{i,(k+1)} = A\delta_{i,k}-(d_i+g_i)B_iu_{i,k} +\mathop{\sum}\nolimits_{j \in N_i} a_{ij}B_ju_{j,k}  -(d_i+g_i)E_iw_{i,k}
+\mathop{\sum}\nolimits_{j \in N_i}a_{ij}E_jw_{j,k}.
\end{align}

\subsubsection{Cooperative graphical game for disturbance rejection}

Similar with the cost function in Reference \zcite{jiao2016multi}, we now define the cost function $J_i:=J_i(\delta_{i,0},u_{i},u_{-i},w_{i},w_{-i})$ of each agent $i$  
\begin{align*} 
   J_i  = &\mathop{\sum}\nolimits_{k=0}^{\infty} r_i(\delta_{i,k}, u_{i,k}, u_{-i,k}, w_{i,k}, w_{-i,k}) \nonumber \\
    = & \mathop{\sum}\nolimits_{k=0}^{\infty} \left(\delta_{i,k}^\top Q_{ii} \delta_{i,k} + u_{i,k}^\top R_{ii} u_{i,k} +\mathop{\sum}\nolimits_{j \in N_i} u_{j,k}^\top R_{i j} u_{j,k}  -\beta^2 w_{i,k}^\top T_{ii} w_{i,k} -\beta^2\mathop{\sum}\nolimits_{j \in N_i}  w_{j,k}^\top T_{ij} w_{j,k}\right)
\end{align*}
where $u_i = \mathop{\{u_{i,l}\}}\nolimits_{l=0}^{\infty}$, $w_i = \mathop{\{w_{i,l}\}}\nolimits_{l=0}^{\infty}, u_{-i} = \mathop{\{u_{-il}\}}\nolimits_{l=0}^{\infty}$, $w_{-i} = \mathop{\{w_{-il}\}}\nolimits_{l=0}^{\infty}, Q_{ii}, R_{ii}, T_{ii} > 0$, $R_{ij}, T_{ij} \geq 0$, $\beta\geq\beta^*\geq0$, $\beta^*$ is the smallest gain that the disturbance attenuation can achieve and $r_i(\delta_{i,k}, u_{i,k}, u_{-i,k}, w_{i,k}, w_{-i,k})$ is the cost at the current step.
 
The value function $V_{i,k} := V_i({\delta_{i,k}})$ of each agent $i$ can be written as the Bellman equation
\begin{align}
V_{i,k}    & =\mathop{\sum}\nolimits_{l=k}^{\infty} r_i(\delta_{i,l}, u_{i,l}, u_{-i,l}, w_{i,l}, w_{-i,l})\nonumber\\ 
&= r_i(\delta_{i,k}, u_{i,k}, u_{-i,k}, w_{i,k}, w_{-i,k})+V_{i,k+1}.\label{e2.2.10} 
\end{align}

The objective of cooperative zero-sum graphical game is to find the optimal control policies $u_i^*$ for the synchronization of MAS under the worst-case disturbance $w_i^*$
\begin{align} \label{e2.2.11} 
\textbf{Objective} : \min\limits_{u_{i} \in \mathbb{R}^p}\space\max\limits_{w_{i} \in \mathbb{R}^q}\space J_{i}(\delta_{i,0},u_{i},u_{-i}^*,w_{i},w_{-i}^*),
\end{align}
where $u_{-i}^*$ and $w_{-i}^*$ are the optimal control and the worst disturbance of the neighbor nodes, respectively.

We now state the following standard assumption on the $H_\infty$ optimal control problem.

\begin{assumption} \label{a3} 
 There exists a solution of \eqref{e2.2.11}, i.e., the optimal solution satisfies the Isaacs condition 
\begin{align*}
\min\limits_{u_i \in \mathbb{R}^p}\space\max\limits_{w_{i} \in \mathbb{R}^q}\space  J_{i}(\delta_{i,0},u_{i},u_{-i}^*,w_{i},w_{-i}^*) = \max\limits_{w_i \in \mathbb{R}^q}\space\min\limits_{u_{i}\in \mathbb{R}^p}\space J_{i}(\delta_{i,0},u_{i},u_{-i}^*,w_{i},w_{-i}^*).
\end{align*}
\end{assumption} 
\wxy{To find the optimal solution to the cooperative zero-sum graphical game, the coupled discrete-time Hamilton-Jacobi-Isaacs (DTHJI) equation is defined as} 
\begin{align} \label{e2.2.13}
 \wxy{H_{i,k}}=&\ \wxy{H_i(\delta_{i,k}, \nabla V_{i,k+1}^*, u_{i,k}^*, u_{-i,k}^*, w_{i,k}^*, w_{-i,k}^*)}\nonumber \\
 = &\ \wxy{r_i(\delta_{i,k}, u_{i,k}^*, u_{-i,k}^*, w_{i,k}^*, w_{-i,k}^*) + \nabla V_{i,k+1}^{*\top} \delta_{i,k+1} = 0},
\end{align}
\wxy{where $\nabla V_{i,k+1}^*=\frac{\partial V_{i,k+1}^*}{\partial \delta_{i,k+1}}$, $V_{i,k}^* = \min\limits_{u_{i,k} \in \mathbb{R}^p}\space\max\limits_{ \{w_{i,k} \in \mathbb{R}^q \}} \space \{ r_i(\delta_{i,k}, u_{i,k}, u_{-i,k}, w_{i,k}, w_{-i,k})+V_{i,k+1}^*\}$}, and the optimal-case of the control policies and the worst-case of the disturbance policies calculated by using the stationary condition \cite{lewis2012optimal} are given by
\begin{align} \label{e2.2.14}
     u_{i,k}^* & =\frac{d_i+g_i}{2}R_{ii}^{-1}B_i^\top\nabla V_{i,k+1}^*, \nonumber  \\
     w_{i,k}^* & =-\frac{d_i+g_i}{2\beta^2}T_{ii}^{-1}E_i^\top\nabla V_{i,k+1}^*.
\end{align}

\begin{definition}\citep*{abouheaf2014multi}\label{d1}
    $\{u_i^*, u_{-i}^*, w_i^*, w_{-i}^*\}$ is the Nash solution of the zero-sum graphical game if 
    \begin{align*}
    J_i(\delta_{i,0},u_i^*,u_{-i}^*,w_i,w_{-i}^*) \leq J_i(\delta_{i,0},u_i^*,u_{-i}^*,w_i^*,w_{-i}^*) \leq  J_i(\delta_{i,0},u_i,u_{-i}^*,w_i^*,w_{-i}^*) \nonumber 
    \end{align*}
    holds for $i = 1, 2, \cdots, N$, where $u_i^* = \mathop{\{u_{i,l}^*\}}\nolimits_{l=0}^{\infty}$, $w_i^* = \mathop{\{w_{i,l}^*\}}\nolimits_{l=0}^{\infty}, u_{-i}^* = \mathop{\{u_{-i,l}^*\}}\nolimits_{l=0}^{\infty}$, $w_{-i}^* = \mathop{\{w_{-i,l}^*\}}\nolimits_{l=0}^{\infty}$.
\end{definition}

\begin{lemma} \label{l3}
Let $V_{i,k}^*$ be a positive definite solution of the DTHJI equation \eqref{e2.2.13}, or, equivalently the Bellman equation \eqref{e2.2.10} with the control policies $u_{i,k}^*$ and $w_{i,k}^*$ given by equation \eqref{e2.2.14}. 
Under Assumption \ref{a1} and Assumption \ref{a3}, the following statements hold:
\begin{itemize}
    \item[\textbf{1)}] \wxy{If $\beta > 0$ is large enough such that the following inequalities
        \begin{align} \label{stable_condition1}
        B_l^\top R_{lj}^{-1} B_l \geq \frac{1}{\beta^2} E_l^\top T_{lj}^{-1} E_l,  \qquad \forall j \in N_l \cup \{ l\}
        \end{align}
    hold for all $l = 1, 2, \dots, N$, then the error dynamics \eqref{e2.2.8} is asymptotically stable, and all agents synchronize to the leader;}
    \item[\textbf{2)}] The optimal cost function of each agent $i$ is given by $$J_i(\delta_{i,0},u_i^*,u_{-i}^*,w_i^*,w_{-i}^*)=V_{i,0}^*;$$
    \item[\textbf{3)}] \wxy{The system \eqref{e2.2.8} is $L_2$ stable with $L_2$-gain bounded by $\beta$;}
    \item[\textbf{4)}] $u_{i}^*$, $u_{-i}^*$, $w_{i}^*$ and $w_{-i}^*$ constitute a Nash equilibrium solution.
\end{itemize}
\end{lemma}

\begin{proof}
\textbf{1)} 
\wxy{To analyze the stability of the closed-loop system, we pose $V_{i,k}^*$ as the Lyapunov function candidate. Its difference at step $k$ is
\begin{align*}  
    \Delta V_{i,k}^* = & -\delta_{i,k}^\top Q_{ii} \delta_{i,k} - u_{i,k}^{*\top} R_{ii} u_{i,k}^* -  \mathop{\sum}\nolimits_{j \in N_i} u_{j,k}^{*\top} R_{i j} u_{j,k}^*  +\beta^2 w_{i,k}^{*\top} T_{ii} w_{i,k}^* +\beta^2\mathop{\sum}\nolimits_{j \in N_i}  w_{j,k}^{*\top} T_{ij} w_{j,k}^* \\
    = & -\delta_{i,k}^\top Q_{ii} \delta_{i,k} - \frac{(d_i + g_i)^2}{4}\nabla V_{i,k+1}^{*\top} \left( B_i^\top R_{ii}^{-1} B_i - \frac{1}{\beta^2} E_i^\top T_{ii}^{-1} E_i\right)\nabla V_{i,k+1}^* \\
    & - \mathop{\sum}\nolimits_{j \in N_i} \frac{a_{ij}^2}{4} \nabla V_{j,k+1}^{*\top} \left( B_j^\top R_{ij}^{-1} B_j - \frac{1}{\beta^2} E_j^\top T_{ij}^{-1} E_j \right) \nabla V_{j,k+1}^*
\end{align*}
From \eqref{stable_condition1}, one has $\Delta V_{i,k}^* \leq  -\underline{\sigma}(Q) \lVert \delta_{i,k} \lVert^2$.
By Lyapunov stability theorem, the closed-loop system \eqref{e2.2.8} is asymptotically stable. 
Using Lemma \ref{l1}, all agents synchronize to the leader.}

\textbf{2)} The cost function of agent $i$ subject to the control policies $u_i$ and $w_i$ can be written as
\begin{align} \label{e2.2.23}
        J_i(\delta_{i,0},u_{i},u_{-i},w_{i},w_{-i}) = &\mathop{\sum}\nolimits_{k=0}^{\infty} r_i(\delta_{i,k}, u_{i,k}, u_{-i,k}, w_{i,k}, w_{-i,k}) \nonumber \\
        =&\mathop{\sum}\nolimits_{k=0}^{\infty} \left(r_i(\delta_{i,k}, u_{i,k}, u_{-i,k}, w_{i,k}, w_{-i,k}) - r_i(\delta_{i,k}, u_{i,k}^*, u_{-i,k}^*, w_{i,k}^*, w_{-i,k}^*)\right) +  V_{i,0}^* 
\end{align}
Letting $u_{-i}$ and $w_{-i}$ adopt the optimal policies given by \eqref{e2.2.14}, one has
\begin{align*} 
    J_i(\delta_{i,0},u_{i},u_{-i}^*,w_{i},w_{-i}^*) \nonumber = &\ V_{i,0}^* + \mathop{\sum}\nolimits_{k = 0}^{\infty}(u_{i,k}-u_{i,k}^*)^\top R_{ii}(u_{i,k}-u_{i,k}^*)+ 2u_{i,k}^{*\top}R_{ii}(u_{i,k}-u_{i,k}^*)\\ &- 2\beta^2w_{i,k}^{*\top} T_{ii}(w_{i,k}-w_{i,k}^*) - \beta^2(w_{i,k}-w_{i,k}^*)^\top T_{ii} (w_{i,k}-w_{i,k}^*).
\end{align*}
Then, taking the optimal control polices $u_{i}^*$ and $w_{i}^*$ yields 
\begin{align*}
    J_i(\delta_{i,0},u_i^*,u_{-i}^*,w_i^*,w_{-i}^*)=V_{i,0}^*.
\end{align*}

\textbf{3)} \wxy{Let $u_{i}$, $u_{-i}$ and $w_{-i}$ take the optimal forms. If Assumption \ref{a3} holds, the cooperative graphical game has a unique solution, and hence the optimal control policy and worst disturbance policy are unique \cite{jiao2016multi}. Then, one has 
\begin{align} \label{po1}
    J_i(\delta_{i,0},u_{i}^*,u_{-i}^*,w_{i}^*,w_{-i}^*) \geq J_i(\delta_{i,0},u_{i}^*,u_{-i}^*,w_{i},w_{-i}^*).
\end{align}
Since $J_i(\delta_{i,0},u_i^*,u_{-i}^*,w_i^*,w_{-i}^*)=V_{i,0}^*$, \eqref{po1} implies that
\begin{align*}
\mathop{\sum}\nolimits_{k=0}^{\infty} r_i(\delta_{i,k}, u_{i,k}^*, u_{-i,k}^*, w_{i,k}, w_{-i,k}^*)\leq V_{i,0}^*,
\end{align*}
and further
$$\mathop{\sum}\nolimits_{k=0}^{\infty} \left(\delta_{i,k}^\top Q_{ii} \delta_{i,k} + u_{i,k}^{*\top} R_{ii} u_{i,k}^* + \mathop{\sum}\nolimits_{j \in N_i} u_{j,k}^{*\top} R_{i j} u_{j,k}^* \right) \leq V^*_{i,0} + \beta^2\mathop{\sum}\nolimits_{k=0}^{\infty}  \left(w_{i,k}^\top T_{ii} w_{i,k} +  \mathop{\sum}\nolimits_{j \in N_i}  w_{j,k}^{*\top} T_{ij} w_{j,k}^* \right).$$
Considering the fact that $w_{i,k} \in L_2[0, \infty)$, $i = 1, 2, \cdots, N$, $L_2$ stability of \eqref{e2.2.8} is achieved.}

\textbf{4)} For arbitrary control policies $u_{i}$, let the initial optimal value function be denoted as $V_{i,0}^*$, and let $u_{-i}$, $w_{i}$ and $w_{-i}$ take the optimal forms.
Then $\eqref{e2.2.23}$ implies
\begin{align*}
    \mathop{\sum}\nolimits_{k=0}^{\infty} r_i(\delta_{i,k}, u_{i,k}, u_{-i,k}^*, w_{i,k}^*, w_{-i,k}^*) + \Delta V_{i,k}^* \geq 0.
\end{align*}
For arbitrary disturbance policies $w_{i}$, setting $u_{i}$, $u_{-i}$, $w_{-i}$ to the optimal form yields
\begin{align*}
    \mathop{\sum}\nolimits_{k=0}^{\infty} r_i(\delta_{i,k}, u_{i,k}^*, u_{-i,k}^*, w_{i,k}, w_{-i,k}^*) + \Delta V_{i,k}^* \leq 0.
\end{align*}
Noting that $\mathop{\sum}\nolimits_{k=0}^{\infty}  \Delta V_{i,k}^* = -J_i(\delta_{i,0},u_i^*,u_{-i}^*,w_i^*,w_{-i}^*)$, one has 
\begin{align*} 
   J_i(\delta_{i,0},u_i^*,u_{-i}^*,w_i,w_{-i}^*)  
    \leq    J_i(\delta_{i,0},u_i^*,u_{-i}^*,w_i^*,w_{-i}^*) \leq  J_i(\delta_{i,0},u_i,u_{-i}^*,w_i^*,w_{-i}^*),
\end{align*}
which shows that $\{{u_{i}^*, u_{-i}^*, w_{i}^*, w_{-i}^*}\}$ constitute a Nash equilibrium solution by Definition \ref{d1}.
\end{proof}

\wxy{
By Lemma \ref{l3}, a Nash equilibrium can be achieved in the zero-sum cooperative graphical game  if the optimal-case of the control policies and the worst-case of the disturbance policies are employed.
However, according to Reference \zcite{lian2024distributed}, the existence of $V_{i,k}$ of the highly-coupled DTHJI equation is not always guaranteed, and the optimal control policy is not fully distributed.
To address these issues, a minmax strategy is employed in the following section that ensures the existence of a solution and solves the optimal control problem in a fully distributed manner.}

\subsubsection{Non-cooperative graphical game for disturbance rejection}
\wxy{
By considering disturbance $w_{i}$, neighbors' control policy $u_{-i}$ and disturbance $w_{-i}$ of agent $i$ as adversarial inputs, agent $i$ makes the best decision for the worst-case behaviors of its neighbors. 
First, we modify the cost function as
\begin{align*}
    \mathcal{J}_{i}(\delta_{i,0},u_{i},u_{-i},w_{i},w_{-i}) = \mathop{\sum}\nolimits_{k=0}^{\infty} \underbrace{\left(\delta_{i,k}^\top Q_{ii} \delta_{i,k} + u_{i,k}^\top R_{ii} u_{i,k} - \Breve{\beta}^2 \mathop{\sum}\nolimits_{j \in N_i} u_{j,k}^\top R_{i j} u_{j,k}  -\Breve{\beta}^2 w_{i,k}^\top T_{ii} w_{i,k} -\Breve{\beta}^2\mathop{\sum}\nolimits_{j \in N_i}  w_{j,k}^\top T_{ij} w_{j,k}\right)}_{\Breve{r}_i(\delta_{i,k}, u_{i,k}, u_{-i,k}, w_{i,k}, w_{-i,k})},
\end{align*}
where $\Breve{\beta}\geq \Breve{\beta}^*\geq0$ and $\Breve{\beta}^*$ is the smallest gain that the disturbance attenuation can achieve.
The minmax strategy employed in this section is defined as follows. 
\begin{definition} \citep*{lian2022online}
    In a non-cooperative graphical game for \eqref{e2.2.8}, the minmax strategy of agent $i$ , $\forall i = 1, 2, \cdots, N$, is defined as 
    \begin{align*}
        u_i^* = \arg \min\limits_{u_{i} \in \mathbb{R}^p}\space\max\limits_{ \{u_{-i} \in \mathbb{R}^p,w_{i}, w_{-i} \in \mathbb{R}^q \}}\space \mathcal{J}_{i}(\delta_{i,0},u_{i},u_{-i},w_{i},w_{-i}).
    \end{align*}
\end{definition}
}

\wxy{
The objective of this section is to seek the distributed minmax solution for the following non-cooperative graphical game 
\begin{align*}
\textbf{Objective} : \min\limits_{u_{i} \in \mathbb{R}^p}\space\max\limits_{ \{u_{-i} \in \mathbb{R}^p,w_{i}, w_{-i} \in \mathbb{R}^q \}} \space \mathcal{J}_{i}(\delta_{i,0},u_{i},u_{-i},w_{i},w_{-i}).
\end{align*}
Correspondingly, we define the value function $\mathcal{V}_{i,k} := \mathcal{V}_i(\delta_{i,k})$ as
\begin{align} \label{non_bellm}
\mathcal{V}_{i,k} = \Breve{r}_i(\delta_{i,k}, u_{i,k}, u_{-i,k}, w_{i,k}, w_{-i,k})+\mathcal{V}_{i,k+1},
\end{align}
and the Hamiltonian function $\mathcal{H}_{i,k} := \mathcal{H}_i(\delta_{i,k}, \nabla \mathcal{V}_{i,,k+1}, u_{i,k}, u_{-i,k}, w_{i,k}, w_{-i,k})$ becomes 
\begin{align} \label{non_Ham}
 \mathcal{H}_{i,k}= \Breve{r}_i(\delta_{i,k}, u_{i,k}, u_{-i,k}, w_{i,k}, w_{-i,k}) + \nabla \mathcal{V}_{i,k+1}^\top \delta_{i,(k+1)},
\end{align}
where $\nabla \mathcal{V}_{i,k+1}=\frac{\partial \mathcal{V}_{i,k+1}}{\partial \delta_{i,k+1}}$.
Taking $\frac{\partial \mathcal{H}_{i,k}}{\partial u_{i,k}} = 0, \frac{\partial \mathcal{H}_{i,k}}{\partial w_{i,k}} = 0$ yields the optimal control policy and the worst-case disturbance
\begin{align} \label{non_ac}
    u_{i,k} & =\frac{d_i+g_i}{2}R_{ii}^{-1}B_i^\top \nabla \mathcal{V}_{i,k+1}, \nonumber  \\
    w_{i,k} & =-\frac{d_i+g_i}{2\Breve{\beta}^2}T_{ii}^{-1}E_i^\top \nabla \mathcal{V}_{i,k+1}.
\end{align}
Similarly, computing $\frac{\partial \mathcal{H}_{i,k}}{\partial u_{j,k}} = 0, \frac{\partial \mathcal{H}_{i,k}}{\partial u_{j,k}} = 0, \; \forall j \in N_i$, yields the worst-case control policy and disturbance of agent $i$'s neighbors
\begin{align} \label{non_adv}
     u_{j,k} & =\frac{a_{ij}}{2\Breve{\beta}^2}R_{ij}^{-1}B_j^\top \nabla \mathcal{V}_{i,k+1}, \nonumber  \\
     w_{j,k} & =\frac{a_{ij}}{2\Breve{\beta}^2}T_{ij}^{-1}E_j^\top \nabla \mathcal{V}_{i,k+1}.
\end{align}
}

\wxy{
Let $\mathcal{V}_{i, k}$ take the optimal value, i.e., $\mathcal{V}_{i, k}^* := \min\limits_{u_{i,k} \in \mathbb{R}^p}\space\max\limits_{  \{u_{-i,k} \in \mathbb{R}^p,w_{i,k}, w_{-i,k} \in \mathbb{R}^q\}} \space \{ \Breve{r}_i(\delta_{i,k}, u_{i,k}, u_{-i,k}, w_{i,k}, w_{-i,k})+ \mathcal{V}_{i,k+1}^*\}$.
Then, substituting \eqref{non_ac} and \eqref{non_adv} into \eqref{non_Ham} yields the discrete-time Hamilton-Jacobi (DTHJ) equation 
\begin{align} \label{non_HJ}
    \Breve{r}_i&\left(\delta_{i,k}, \frac{d_i+g_i}{2}R_{ii}^{-1}B_i^\top \nabla \mathcal{V}_{i,k+1}^*, \frac{a_{ij}}{2\Breve{\beta}^2}R_{ij}^{-1}B_j^\top \nabla \mathcal{V}_{i,k+1}^*, -\frac{d_i+g_i}{2\Breve{\beta}^2}T_{ii}^{-1}E_i^\top \nabla \mathcal{V}_{i,k+1}^*, \frac{a_{ij}}{2\Breve{\beta}^2}T_{ij}^{-1}E_j^\top \nabla \mathcal{V}_{i,k+1}^*\right) \nonumber \\
    & + \nabla \mathcal{V}_{i,k+1}^{*\top} (A\delta_{i,k}-\frac{(d_i+g_i)^2}{2}B_iR_{ii}^{-1}B_i^\top \nabla \mathcal{V}_{i,k+1}^* +\mathop{\sum}\nolimits_{j \in N_i} \frac{a_{ij}^2}{2\Breve{\beta}^2}B_jR_{ij}^{-1}B_j^\top \nabla \mathcal{V}_{i,k+1}^* \nonumber \\
    & -\frac{(d_i+g_i)^2}{2\Breve{\beta}^2}E_iT_{ii}^{-1}E_i^\top \nabla \mathcal{V}_{i,k+1}^*+\mathop{\sum}\nolimits_{j \in N_i}\frac{a_{ij}^2}{2\Breve{\beta}^2}E_jT_{ij}^{-1}E_j^\top \nabla \mathcal{V}_{i,k+1}^*) = 0.
\end{align}
\begin{lemma} \label{l4}
Let $\mathcal{V}_{i,k}^*$ be a positive definite solution of the DTHJ equation \eqref{non_HJ}, or, equivalently the Bellman equation \eqref{non_bellm} with the control policies $u_{i}^*$, $u_{-i}^*$, $w_{i}^*$ and $w_{-i}^*$ given by equation \eqref{non_ac}, \eqref{non_adv}. 
Under Assumption \ref{a1}, the following statements hold:
\begin{itemize}
    \item[\textbf{1)}] If $\Breve{\beta} > 0$ is large enough such that
        \begin{align} \label{stable_condition2}
        (d_i + g_i)^2 B_i^\top R_{ii}^{-1} B_i  \geq  \frac{(d_i + g_i)^2}{\Breve{\beta}^2} E_i^\top T_{ii}^{-1} E_i + \mathop{\sum}\nolimits_{j \in N_i} \frac{a_{ij}^2}{\Breve{\beta}^2} \left( B_j^\top R_{ij}^{-1} B_j +  E_j^\top T_{ij}^{-1} E_j \right),
        \end{align}
    then the error dynamics \eqref{e2.2.8} is asymptotically stable, and all agents synchronize to the leader node;
    \item[\textbf{2)}] The optimal cost function of each agent $i$ is given by $$\mathcal{J}_i(\delta_{i,0},u_i^*,u_{-i}^*,w_i^*,w_{-i}^*)=\mathcal{V}_{i,0}^*;$$
    \item[\textbf{3)}] The system \eqref{e2.2.8} is $L_2$ stable with $L_2$-gain bounded by $\Breve{\beta}$;
    \item[\textbf{4)}] The minmax strategy $u_{i}^*$ is fully distributed, and not a Nash equilibrium.
\end{itemize}
\end{lemma}
}

\begin{proof}
\textbf{1)} 
\wxy{
To analyze the stability of the closed-loop system, we take $\mathcal{V}_{ik}^*$ as the Lyapunov function candidate and its difference at step $k$ is
\begin{align*}  
    \Delta \mathcal{V}_{i,k}^* = & -\delta_{i,k}^\top Q_{ii} \delta_{i,k} - u_{i,k}^{*\top} R_{ii} u_{i,k}^* + \Breve{\beta}^2 \mathop{\sum}\nolimits_{j \in N_i} u_{j,k}^{*\top} R_{i j} u_{j,k}^*  + \Breve{\beta}^2 w_{i,k}^{*\top} T_{ii} w_{i,k}^* + \Breve{\beta}^2\mathop{\sum}\nolimits_{j \in N_i}  w_{j,k}^{*\top} T_{ij} w_{j,k}^* \\
    = & -\delta_{i,k}^\top Q_{ii} \delta_{i,k} - \frac{1}{4}\nabla \mathcal{V}_{i,k+1}^{*\top} \left( (d_i + g_i)^2 B_i^\top R_{ii}^{-1} B_i  - \mathop{\sum}\nolimits_{j \in N_i} \frac{a_{ij}^2}{\Breve{\beta}^2}  B_j^\top R_{ij}^{-1} B_j  - \frac{(d_i + g_i)^2}{\Breve{\beta}^2} E_i^\top T_{ii}^{-1} E_i \right. \\
    & \left. - \mathop{\sum}\nolimits_{j \in N_i} \frac{a_{ij}^2}{\Breve{\beta}^2} E_j^\top T_{ij}^{-1} E_j \right) \nabla \mathcal{V}_{i,k+1}^*
\end{align*}
From \eqref{stable_condition2}, one has $\Delta \mathcal{V}_{i,k}^* \leq  -\underline{\sigma}(Q) \lVert \delta_{i,k} \lVert^2$.
By Lyapunov stability theorem, the closed-loop system \eqref{e2.2.8} is asymptotically stable. 
Using Lemma \ref{l1}, all agents synchronize to the leader.
}

\textbf{2)} 
\wxy{
The cost function of agent $i$ subject to the control policies $u_i$ and $w_i$ can be written as
\begin{align*}
        \mathcal{J}_i(\delta_{i,0},u_{i},u_{-i},w_{i},w_{-i}) = &\mathop{\sum}\nolimits_{k=0}^{\infty} \Breve{r}_i(\delta_{i,k}, u_{i,k}, u_{-i,k}, w_{i,k}, w_{-i,k}) \nonumber \\
        =&\mathop{\sum}\nolimits_{k=0}^{\infty} \left(\Breve{r}_i(\delta_{i,k}, u_{i,k}, u_{-i,k}, w_{i,k}, w_{-i,k}) - \Breve{r}_i(\delta_{i,k}, u_{i,k}^*, u_{-i,k}^*, w_{i,k}^*, w_{-i,k}^*)\right) +  \mathcal{V}_{i,0}^* .
\end{align*}
Let $u_{-i}$ and $w_{-i}$ adopt the optimal form \eqref{non_adv}. One has
\begin{align*} 
    \mathcal{J}_i(\delta_{i,0},u_{i},u_{-i}^*,w_{i},w_{-i}^*) \nonumber = & \mathcal{V}_{i,0}^* + \mathop{\sum}\nolimits_{k = 0}^{\infty}(u_{i,k}-u_{i,k}^*)^\top R_{ii}(u_{i,k}-u_{i,k}^*)+ 2u_{i,k}^{*\top}R_{ii}(u_{i,k}-u_{i,k}^*)\\ &- 2\Breve{\beta}^2w_{i,k}^{*\top} T_{ii}(w_{i,k}-w_{i,k}^*) - \Breve{\beta}^2(w_{i,k}-w_{i,k}^*)^\top T_{ii} (w_{i,k}-w_{i,k}^*).
\end{align*}
Then, taking the optimal polices \eqref{non_ac} yields 
\begin{align*}
    \mathcal{J}_i(\delta_{i,0},u_i^*,u_{-i}^*,w_i^*,w_{-i}^*)=\mathcal{V}_{i,0}^*.
\end{align*}
}

\textbf{3)} 
\wxy{
Since $V^*_{i,0}$ is positive definite and 
\begin{align*}
     \min\limits_{u_{i} \in \mathbb{R}^p}\space\max\limits_{ \{u_{-i} \in \mathbb{R}^p,w_{i}, w_{-i} \in \mathbb{R}^q \}} \space \mathcal{J}_{i}(\delta_{i,0},u_{i},u_{-i},w_{i},w_{-i}) \geq  \min\limits_{u_{i} \in \mathbb{R}^p}\space \mathcal{J}_{i}(\delta_{i,0},u_{i},u_{-i},w_{i},w_{-i}), 
\end{align*}
one has $\min\limits_{u_{i} \in \mathbb{R}^p}\space \mathcal{J}_{i}(\delta_{i,0},u_{i},u_{-i},w_{i},w_{-i}) \leq V^*_{i,0}$, i.e.,
\begin{align*}
\mathop{\sum}\nolimits_{k=0}^{\infty} \Breve{r}_i(\delta_{i,k}, u_{i,k}^*, u_{-i,k}, w_{i,k}, w_{-i,k})\leq \mathcal{V}_{i,0}^*.
\end{align*}
And the $L_2$ stability of \eqref{e2.2.8} can be shown by
$$\mathop{\sum}\nolimits_{k=0}^{\infty} \left(\delta_{i,k}^\top Q_{ii} \delta_{i,k} + u_{i,k}^{*\top} R_{ii} u_{i,k}^* \right) \leq \mathcal{V}^*_{i,0} + \Breve{\beta}^2\mathop{\sum}\nolimits_{k=0}^{\infty}  \left(\mathop{\sum}\nolimits_{j \in N_i} u_{j,k}^\top R_{i j} u_{j,k}  +  w_{i,k}^\top T_{ii} w_{i,k} +  \mathop{\sum}\nolimits_{j \in N_i}  w_{j,k}^\top T_{ij} w_{j,k} \right).$$
}

\textbf{4)} 
\wxy{By comparing \eqref{e2.2.13} and \eqref{non_HJ}, it is clear that the minmax strategy $u_{i}^*$ is fully distributed.
As shown in Reference \zcite{lian2022online}, the minmax strategy is not Nash in the non-cooperative graphical game since $u_{-i}$ of agent $i$'s neighbors in the non-cooperative game are to maximize the cost function $\mathcal{J}_i$, while the $u_{-i}$ of the Nash equilibrium solution are to minimize it.
}
\end{proof}

\section{Solution for graphical game using Q-function}\label{Sec PIZSG}
From \eqref{e2.2.14}, it is noted that the optimal policies for the zero-sum graphical game require knowledge of the system dynamics of each agent.
However, most real-world control systems are nonlinear, with uncertain parameters, whose precise models are difficult to obtain. 
If the system dynamics are unknown, the optimal policies cannot be adopted to solve the $H_{\infty}$ optimal control problem for MAS. 
In this section, a Q-function-based PI algorithm is proposed to solve the $H_{\infty}$ optimal problem for MAS, where the overall system dynamics are assumed to be unknown. 
\subsection{Q-function-based PI algorithm for cooperative graphical game}
According to Reference \zcite{zhang2016data}, the local Q-function of each agent $i$ is defined as 
\begin{align*} 
Q_i(\delta_{i,k}, u_{i,k}, u_{-i,k}, w_{i,k}, w_{-i,k})  =   r_i(\delta_{i,k}, u_{i,k}, u_{-i,k}, w_{i,k}, w_{-i,k}) + V_i(\delta_{i,k+1}).
\end{align*}
Compared with the value function \eqref{e2.2.10}, this Q-function contains not only the local neighbourhood tracking error but also the control and disturbance input.
Denote $Q_i(\delta_{i,k}, u_{i,k}, u_{-i,k}, w_{i,k}, w_{-i,k})$ as $Q_{i,k}$ for brevity. 
Noting that $Q_{i,k} = V_i(\delta_{i,k})$, the DTHJI equation based on Q-function can be written as
\begin{align*}
Q_{i,k}  = r_i(\delta_{i,k}, u_{i,k}, u_{-i,k}, w_{i,k}, w_{-i,k}) + Q_{i,k+1}.
\end{align*}

To solve the DTHJI equation without any information on system dynamics, a Q-function-based PI method is proposed as follows. Let $\epsilon_1$ and $\epsilon_2$ be small positive values acting as tolerances to stop the loop in the following algorithm.

\DontPrintSemicolon
\IncMargin{1em}
\begin{algorithm} 
\SetKwData{Left}{left}
\SetKwData{This}{this}
\SetKwData{Up}{up} 
\SetKwFunction{Union}{Union}
\SetKwFunction{FindCompress}{FindCompress} 
\SetKwInOut{Input}{Step 1}
\SetKwInOut{Output}{Step 2}
\Input{\textbf{Initialize} admissible control policies} 
\Output{}
\For{$u_{i,k}^h, \forall h = 0,1,\dots$, \textnormal{at each step} $h$,}{ 
        \Repeat{$\lVert \mathop{\textnormal{min}}\limits_{u_{i}}Q_{i}^{h,z} - Q_i^{h,z} \rVert \leq \epsilon_2$}{ 
            \For{$w_{i,k}^z, \forall z = 0,1,\dots$, \textnormal{at each step} $z$, }{
            \Repeat{$\lVert \mathop{\textnormal{max}}\limits_{w_{i}}Q_{i}^{h,z} - Q_{i}^{h,z} \rVert \leq \epsilon_1$
            }{
        \textbf{Solve} the following equation for $Q_i^{h,z}$        
           \begin{flalign}\label{e3.1.3}
            Q_{i}^{h,z}\left(\delta_{i,k},u_{i,k}^h,u_{-i,k}^h,w_{i,k}^z,w_{-i,k}^z\right)
            = & r_i\left(\delta_{i,k},u_{i,k}^h,u_{-i,k}^h,w_{i,k}^z,w_{-i,k}^z\right) \nonumber \\
            & + Q_i^{h,z}\big(\delta_{i,k+1},u_{i,k+1}^h, u_{-i,k+1}^h, w_{i,k+1}^z,w_{-i,k+1}^z\big)
            \end{flalign} 
            
            \textbf{Update} the disturbance policies
            \begin{align}\label{e3.1.4}
            w_{i}^{z+1} = \mathop{\textnormal{argmax}}\limits_{w_{i}}Q_{i}^{h,z}\left(\delta_{i,k},u_{i,k}^h,u_{-i,k}^h,w_{i},w_{-i,k}^z\right)
            \end{align}
 	 	   }
            }
        \textbf{Update} the control policies
        \begin{align}\label{e3.1.5}
        u_{i}^{h+1}  =\mathop{\textnormal{argmin}}\limits_{u_{i}}Q_{i}^{h,z}\left(\delta_{i,k},u_{i},u_{-i,k}^h,w_{i,k}^z,w_{-i,k}^z\right)
        \end{align}
 	}
        } 
 	   \caption{\wxy{PI Algorithm for cooperative graphical game for disturbance rejection}}
 	 	  \label{algo_1} 
 	 \end{algorithm}
 \DecMargin{1em}

\begin{remark}\label{rem1}
    It should be noted that $A$, $B_i$, and $E_i$ need to be known for the HJI equation \eqref{e2.2.13} and the optimal policies \eqref{e2.2.14} by the traditional PI algorithm (cf. Reference \zcite{jiao2016multi}), while the Q-function based PI algorithm is model-free since the control policies and disturbance policies are implicitly contained in the Q-function. 
\end{remark}

Define relative weights $\rho_{ij} = \bar{\sigma}(T_{jj}^{-1}T_{ij})$ and $\kappa_{ij} = \bar{\sigma}({R_{jj}^{-1}R_{ij}})$ for $\forall i = 1, 2, \cdots, N$ and $\forall j \in N_i$.
The next theorem confirms the convergence of Algorithm 1 when all agents are updated simultaneously.
\begin{thm} \label{thm1}
    Given arbitrary initial admissible control policies, let all follower agents perform Algorithm \ref{algo_1} simultaneously. 
    Then $u_i$, $w_i$, $u_{-i}$ and $w_{-i}$ converge to the Nash equilibrium and the Q-function converges to the optimal solution of DTHJI equation for large weights $R_{ii}, T_{ii}$ and small relative weights $\rho_{ij}, \kappa_{ij}$.
\end{thm}

\begin{proof}
    \wxy{
    The proof is carried out in two steps. 
    First, the disturbance policies are updated with the control policies fixed in the inner loop. 
    Define a new function as 
    \begin{align*}
        U_{i}(\delta_{i,k}, u_{i,k}^h,  u_{-i,k}^h, w_{i,k}^{z+1},w_{-i,k}^z) = r_i(\delta_{i,k}, u_{i,k}^h, u_{-i,k}^h, w_{i,k}^{z+1}, w_{-i,k}^z) + Q_{i}^{h,z}(\delta_{i,k+1},u_{i,k+1}^h,u_{-i,k+1}^h, w_{i,k+1}^{z}, w_{-i,k+1}^z).
    \end{align*}
    According to \eqref{e3.1.4}, one has
        \wxy{
        \begin{align} \label{e3.2.1}
          Q_{i}^{h,z}(\delta_{i,k},u_{i,k}^h,u_{-i,k}^h, w_{i,k}^{z}, w_{-i,k}^z)\leq  U_i(\delta_{i,k},u_{i,k}^h,u_{-i,k}^h,w_{i,k}^{z+1},w_{-i,k}^z).
        \end{align}
        }
    From \eqref{e3.1.3}, we have
        \wxy{
        \begin{align*} 
         U_i(\delta_{i,k},u_{i,k}^h,u_{-i,k}^h,w_{i,k}^{z+1},w_{-i,k}^z) = Q_{i}^{h,z+1}(\delta_{i,k},u_{i,k}^h,u_{-i,k}^h,w_{i,k}^{z+1},w_{-i,k}^{z+1})  + \Delta r(w_{-i,k}^{z},w_{-i,k}^{z+1}),
        \end{align*}
        }
    where 
        \wxy{
        \begin{align*}
          \Delta r (w_{-i,k}^{z}, w_{-i,k}^{z+1}) 
        = & r_i(\delta_{i,k}, u_{i,k}^h, u_{-i,k}^h, w_{i,k}^{z+1}, w_{-i,k}^z) - r_i(\delta_{i,k}, u_{i,k}^h, u_{-i,k}^h, w_{i,k}^{z+1}, w_{-i,k}^{z+1})\\
        & + \mathop{\sum}\nolimits_{l=k+1}^{\infty} \left(r_i(\delta_{i,l}, u_{i,l}^h, u_{-i,l}^h, w_{i,l}^z, w_{-i,l}^z) - r_i(\delta_{i,l}, u_{i,l}^h, u_{-i,l}^h, w_{i,l}^{z+1}, w_{-i,l}^{z+1}) \right)\\
        = & \beta^2 \left(\mathop{\sum}\nolimits_{l=k+1}^{\infty} \left( w_{i,l}^{z+1 \top} T_{ii} w_{i,l}^{z+1} - w_{i,l}^{z \top} T_{ii} w_{i,l}^{z} \right) + \mathop{\sum}\nolimits_{l=k}^{\infty}\mathop{\sum}\nolimits_{j \in N_i}  \left( w_{j,l}^{z+1 \top} T_{ij} w_{j,l}^{z+1} - w_{j,l}^{z \top} T_{ij} w_{j,l}^{z} \right) \right) \\
        = & -\beta^2 \mathop{\sum}\nolimits_{l=k}^{\infty} \mathop{\sum}\nolimits_{j \in N_i} \left( 2 w_{j,l}^{z+1 \top} T_{ij}(w_{j,l}^{z} - w_{j,l}^{z+1})  + (w_{j,l}^{z} - w_{j,l}^{z+1 })^{\top} T_{ij} (w_{j,l}^{z} - w_{j,l}^{z+1}) \right) \\
        & - \beta^2 \mathop{\sum}\nolimits_{l=k+1}^{\infty} \left( 2 w_{i,l}^{z+1 \top} T_{ii}(w_{i,l}^{z} - w_{i,l}^{z+1})  + (w_{i,l}^{z} - w_{i,l}^{z+1 })^{\top} T_{ii} (w_{i,l}^{z} - w_{i,l}^{z+1}) \right)
        \end{align*}
        }
    To guarantee the following inequality
    \wxy{
    \begin{align}\label{e3.2.2}
          U_i(\delta_{i,k},u_{i,k}^h,u_{-i,k}^h,w_{i,k}^{z+1},w_{-i,k}^z)  \leq  Q_{i}^{h,z+1}(\delta_{i,k},u_{i,k}^h,u_{-i,k}^h,w_{i,k}^{z+1},w_{-i,k}^{z+1}),
    \end{align}
    }
    a sufficient condition is needed by considering \eqref{e2.2.14}
    \wxy{
    \begin{align} \label{lm_ineq1}
        \beta^2\underline{\sigma}(T_{ij})\lVert \Delta w_{j,l}^z \rVert & \geq (d_j+g_j)  \rho_{ij} \lVert E_j \rVert \lVert \nabla V_{j,l+1}^{h,z} \rVert,
    \end{align}
    }
    $\forall l \geq k$ and $j \in N_i \cup \{ i\} $, where \wxy{$\Delta w_{j,l}^z = w_{j,l}^{z+1} - w_{j,l}^{z}$} and $V_{j,l+1}^{h,z} = Q_{i}^{h,z}(\delta_{i,l+1},u_{i,l+1}^{h},u_{-i,l+1}^{h},w_{i,l+1}^{z},w_{-i,l+1}^{z})$. 
    It can be seen obviously that the inequality \eqref{lm_ineq1} holds for \wxy{small values of $\rho_{ij}$ and large values of $\underline{\sigma}(T_{ii})$.} 
    Combining \eqref{e3.2.1} and \eqref{e3.2.2}, we have
    \wxy{
    \begin{align*}
        Q_i^{h,z+1}(\delta_{i,k}, u_{i,k}, u_{-i,k}, w_{i,k}, w_{-i,k}) \geq  Q_i^{h,z}(\delta_{i,k}, u_{i,k}, u_{-i,k}, w_{i,k}, w_{-i,k}).
    \end{align*}  
    }
Next, consider the outer loop. Define a new function as 
    \begin{align*}
        Y_{i}(\delta_{i,k}, u_{i,k}^{h+1},  u_{-i,k}^h, w_{i,k}^{z},w_{-i,k}^z) = r_i(\delta_{i,k}, u_{i,k}^{h+1}, u_{-i,k}^h, w_{i,k}^{z}, w_{-i,k}^z) + Q_{i}^{h,z}(\delta_{i,k+1},u_{i,k+1}^h,u_{-i,k+1}^h, w_{i,k+1}^{z}, w_{-i,k+1}^z).
    \end{align*}
According to \eqref{e3.1.5}, one has
    \wxy{
    \begin{align}\label{e3.2.5}
         Q_{i}^{h,z}(\delta_{i,k},u_{i,k}^h,u_{-i,k}^h,w_{i,k}^{z},w_{-i,k}^z) \geq  Y_{i}(\delta_{i,k},u_{i,k}^{h+1},u_{-i,k}^h,w_{i,k}^{z},w_{-i,k}^z).
    \end{align}
    }
    Similar to \eqref{e3.2.2}
    \wxy{
    \begin{align*}
         Y_{i}( \delta_{i,k}, u_{i,k}^{h+1}, u_{-i,k}^h,  w_{i,k}^{z},w_{-i,k}^z) = Q_{i}^{h+1,z}(\delta_{i,k},u_{i,k}^{h+1},u_{-i,k}^{h+1},w_{i,k}^{z},w_{-i,k}^{z})  + \Delta r(u_{-i,k}^{h},u_{-i,k}^{h+1}),
    \end{align*}
    }
    where 
        \wxy{
        \begin{align*}
          \Delta r(u_{-i,k}^{h},u_{-i,k}^{h+1}) 
        = & r_i(\delta_{i,k}, u_{i,k}^{h+1}, u_{-i,k}^h, w_{i,k}^{z}, w_{-i,k}^z) - r_i(\delta_{i,k}, u_{i,k}^{h+1}, u_{-i,k}^{h+1}, w_{i,k}^{z}, w_{-i,k}^{z})\\
        & + \mathop{\sum}\nolimits_{l=k+1}^{\infty} \left(r_i(\delta_{i,l}, u_{i,l}^h, u_{-i,l}^h, w_{i,l}^z, w_{-i,l}^z) - r_i(\delta_{i,l}, u_{i,l}^{h+1}, u_{-i,l}^{h+1}, w_{i,l}^{z}, w_{-i,l}^{z}) \right)\\
        = & - \left(\mathop{\sum}\nolimits_{l=k+1}^{\infty} \left( u_{i,l}^{h+1 \top} R_{ii} u_{i,l}^{h+1} - u_{i,l}^{h \top} R_{ii} u_{i,l}^{h} \right) + \mathop{\sum}\nolimits_{l=k}^{\infty}\mathop{\sum}\nolimits_{j \in N_i}  \left( u_{j,l}^{h+1 \top} R_{ij} u_{j,l}^{h+1} - u_{j,l}^{h \top} R_{ij} u_{j,l}^{h} \right) \right) \\
        = & \mathop{\sum}\nolimits_{l=k}^{\infty} \mathop{\sum}\nolimits_{j \in N_i} \left( 2 u_{j,l}^{h+1 \top} R_{ij}(u_{j,l}^{h} - u_{j,l}^{h+1})  + (u_{j,l}^{h} - u_{j,l}^{h+1 })^{\top} R_{ij} (u_{j,l}^{h} - u_{j,l}^{h+1}) \right) \\
        & +\mathop{\sum}\nolimits_{l=k+1}^{\infty} \left( 2 u_{i,l}^{h+1 \top} R_{ii}(u_{i,l}^{h} - u_{i,l}^{h+1})  + (u_{i,l}^{h} - u_{i,l}^{h+1 })^{\top} R_{ii} (u_{i,l}^{h} - u_{i,l}^{h+1}) \right)
        \end{align*}
        }
    To guarantee the following inequality
    \wxy{
    \begin{align} \label{e3.2.7}
         Y_i(\delta_{i,k},u_{i,k}^{h+1},u_{-i,k}^h,w_{i,k}^{z},w_{-i,k}^z)  \geq  Q_{i}^{h+1,z}(\delta_{i,k},u_{i,k}^{h+1},u_{-i,k}^{h+1},w_{i,k}^{z},w_{-i,k}^{z})
    \end{align}
    }
     a sufficient condition is required by \eqref{e2.2.14}
        \wxy{
        \begin{align}\label{lm_ineq2}
        \underline{\sigma}(R_{ij})\lVert \Delta u_{j,l}^h \rVert & \geq (d_j+g_j) \kappa_{ij} \lVert B_j \rVert \lVert \nabla V_{j,l+1}^{h,z} \rVert
        \end{align}
        }
    for $\forall l \geq k$ and $j \in N_i \cup \{ i\}$, where \wxy{$\Delta u_{j,l}^h = u_{j,l}^{h+1} - u_{j,l}^{h}$}.
    The inequality \eqref{lm_ineq2} can be satisfied for small $\kappa_{ij}$ and large $\underline{\sigma}(R_{ii})$. 
    Combining \eqref{e3.2.5} and \eqref{e3.2.7}, we have
    \wxy{
    \begin{align*}
         Q_i^{h+1,z}(\delta_{i,k}, u_{i,k}, u_{-i,k}, w_{i,k}, w_{-i,k})\leq  Q_i^{h,z}(\delta_{i,k}, u_{i,k}, u_{-i,k}, w_{i,k}, w_{-i,k}).
    \end{align*}
    }Since the Q-function has an upper bound, $Q_{i}^{h,z}$ will converge as the $w_i$ is updated. 
    Additionally, while the disturbance policy $w_i$ is fixed, the Q-function is monotonically decreasing and bounded from below by updating $u_i$. 
    According to the convergence of the Q-function with the updates of $u_i$ and $w_i$, the Q-function will converge to the unique solution of the DTHJI equation.
    This completes the proof.} 
\end{proof}

\subsection{Q-function-based PI algorithm for non-cooperative graphical game}
\wxy{
The local Q-function of each agent $i$ in a non-cooperative graphical game for disturbance rejection is defined as 
\begin{align*} 
       \mathcal{Q}_i(\delta_{i,k}, u_{i,k}, u_{-i,k}, w_{i,k}, w_{-i,k})  =   \Breve{r}_i(\delta_{i,k}, u_{i,k}, u_{-i,k}, w_{i,k}, w_{-i,k}) 
      + \mathcal{V}_i(\delta_{i,k+1}).
\end{align*}
Compared with the value function \eqref{non_bellm}, this Q-function contains not only the local neighbourhood tracking error but also the control and disturbance input.
Denote $\mathcal{Q}_i(\delta_{i,k}, u_{i,k}, u_{-i,k}, w_{i,k}, w_{-i,k})$ as $\mathcal{Q}_{i,k}$ for brevity. Noting that $\mathcal{Q}_{i,k} = \mathcal{V}_i(\delta_{i,k})$, the DTHJ equation based on Q-function can be 
written as
\begin{align*}
     \mathcal{Q}_{i,k}  = \Breve{r}_i(\delta_{i,k}, u_{i,k}, u_{-i,k}, w_{i,k}, w_{-i,k}) 
      + \mathcal{Q}_{i,k+1}.
\end{align*}
To solve the DTHJ equation without any information on system dynamics, a Q-function-based PI method is proposed as follows. Let $\Breve{\epsilon}_1$ and $\Breve{\epsilon}_2$ be small positive values acting as tolerances to stop the loop in the following algorithm.}

\DontPrintSemicolon
\IncMargin{1em}
\begin{algorithm} 
\SetKwData{Left}{left}
\SetKwData{This}{this}
\SetKwData{Up}{up} 
\SetKwFunction{Union}{Union}
\SetKwFunction{FindCompress}{FindCompress} 
\SetKwInOut{Input}{Step 1}
\SetKwInOut{Output}{Step 2}
\Input{\wxy{\textbf{Initialize} admissible control policies}} 
\Output{}
\For{\wxy{$u_{i,k}^h, \forall h = 0,1,\dots$, \textnormal{at each step} $h$},}{ 
        \Repeat{\wxy{$\lVert \mathop{\textnormal{min}}\limits_{u_{i}} \mathcal{Q}_{i}^{h,z} - \mathcal{Q}_i^{h,z} \rVert \leq \Breve{\epsilon}_2$}}{ 
            \For{\wxy{$u_{-i,k}^z, w_{i,k}^z, w_{-i,k}^z \forall z = 0,1,\dots$, \textnormal{at each step} $z$, }}{
            \Repeat{\wxy{$\lVert \mathop{\textnormal{max}}\limits_{\{u_{-i}, w_{i}, w_{-i}\}} \mathcal{Q}_{i}^{h,z} -  \mathcal{Q}_{i}^{h,z} \rVert \leq \Breve{\epsilon}_1$}
            }{\wxy{
        \textbf{Solve} the following equation for $\mathcal{Q}_i^{h,z}$        
           \begin{flalign*}
            \mathcal{Q}_{i}^{h,z}\left(\delta_{i,k},u_{i,k}^h,u_{-i,k}^z,w_{i,k}^z,w_{-i,k}^z\right)
            = & \Breve{r}_i\left(\delta_{i,k},u_{i,k}^h,u_{-i,k}^z,w_{i,k}^z,w_{-i,k}^z\right) \nonumber \\
            & + \mathcal{Q}_i^{h,z}\big(\delta_{i,k+1},u_{i,k+1}^h, u_{-i,k+1}^z, w_{i,k+1}^z,w_{-i,k+1}^z\big)
            \end{flalign*} }
            
            \wxy{\textbf{Update} the disturbance and adversarial policies
            \begin{align*}
            & u_{-i}^{z+1} = \mathop{\textnormal{argmax}}\limits_{u_{-i}} \mathcal{Q}_{i}^{h,z}\left(\delta_{i,k},u_{i,k}^h,u_{-i}^z,w_{i,k}^z,w_{-i,k}^z\right) \\
            & w_{i}^{z+1} = \mathop{\textnormal{argmax}}\limits_{w_{i}}\mathcal{Q}_{i}^{h,z}\left(\delta_{i,k},u_{i,k}^h,u_{-i,k}^z,w_{i},w_{-i,k}^z\right) \\
            & w_{-i}^{z+1} = \mathop{\textnormal{argmax}}\limits_{w_{-i}} \mathcal{Q}_{i}^{h,z}\left(\delta_{i,k},u_{i,k}^h,u_{-i,k}^z,w_{i,k}^z,w_{-i}\right)
            \end{align*}
 	 	   }}
            }
        \wxy{\textbf{Update} the control policies
        \begin{align*}
        u_{i}^{h+1}  =\mathop{\textnormal{argmin}}\limits_{u_{i}} \mathcal{Q}_{i}^{h,z}\left(\delta_{i,k},u_{i},u_{-i,k}^z,w_{i,k}^z,w_{-i,k}^z\right)
        \end{align*}
 	}}
        } 
 	   \caption{\wxy{PI Algorithm for non-cooperative graphical game for disturbance rejection}}
 	 	  \label{algo_2} 
 	 \end{algorithm}
 \DecMargin{1em} 

\wxy{
The next theorem provides the sufficient condition for the convergence of Algorithm \ref{algo_2} when all agents are updated simultaneously.
\begin{thm} \label{thm2}
    Given arbitrary initial admissible control policies, let all follower agents perform Algorithm \ref{algo_2} simultaneously. 
    Then $u_i$ converges to the distributed minmax solution and the Q-function converges to the optimal solution of DTHJ equation for large weights $R_{ii}, R_{ij}, T_{ii}, T_{ij}$, $\forall j \in N_i$.
\end{thm}
\begin{proof}
    The proof is similar to the proof of Theorem \ref{thm1} and thus is omitted. 
\end{proof}}

\section{Online implementation}\label{Sec MFANN}
\wxy{
In this section, the Q-function-based PI methods for cooperative graphical game and non-cooperative graphical game are respectively implemented via an actor-disturber-critic structure and an actor-disturber-adversary-critic structure for online training.
}
\subsection{Online implementation for cooperative graphical game}
\subsubsection{Optimal policies for Q-function}
To implement the online model-free PI method, the form of the optimal policies can be used as the target policies to update the actor and disturber.
Considering the quadratic form of the Q-function, let $S_i$ be the kernel matrix such that
\begin{align*}
    Q_{i,k}  =&\ z_{i,k}^\top S_iz_{i,k} \\
    = &\ z_{i,k}^\top
    \begin{bmatrix}
        S_{\delta_i \delta_i} & S_{\delta_i u_i} & S_{\delta_i u_{-i}} & S_{\delta_i w_i} & S_{\delta_i w_{-i}} \\ 
        S_{u_i \delta_i} & S_{u_i u_i} & S_{u_i u_{-i}} & S_{u_i w_i} & S_{u_i w_{-i}} \\
        S_{u_{-i} \delta_i} & S_{u_{-i} u_i} & S_{u_{-i} u_{-i}} & S_{u_{-i} w_i} & S_{u_{-i} w_{-i}}\\
        S_{w_i \delta_i} & S_{w_i u_i} & S_{w_i u_{-i}} & S_{w_i w_i} & S_{w_i w_{-i}}\\
        S_{w_{-i} \delta_i} & S_{w_{-i} u_i} & S_{w_{-i} u_{-i}} & S_{w_{-i} w_i} & S_{w_{-i} w_{-i}}
    \end{bmatrix}
        z_{i,k},
\end{align*}
where $z_{i,k} = \col(\delta_{i,k},u_{i,k},u_{-i,k},w_{i,k},w_{-i,k})$.

Using the optimal condition $\frac{\partial Q_{i,k}}{\partial 
 u_{i,k}}=0$ and $\frac{\partial Q_{i,k}}{\partial w_{i,k}}=0$, $u_{i,k}^*$ and $w_{i,k}^*$ can be calculated by
\begin{align}\label{e4.1.2}
    u_{i,k}^*  =  & -\left(S_{u_iu_i}-S_{u_iw_i}S_{w_iw_i}^{-1}S_{w_iu_i}\right)^{-1} \left[\left(-S_{u_iw_i}S_{w_iw_i}^{-1}S_{w_i\delta_i}+S_{u_i\delta_i} \right)\delta_{i,k} \nonumber \right. \\  & \left.  +\left(S_{u_iu_{-i}}-S_{u_iw_i}S_{w_iw_i}^{-1}S_{w_iu_{-i}}\right)u_{-i,k}  + \left(S_{u_iw_{-i}}-S_{u_iw_i}S_{w_iw_i}^{-1}S_{w_iw_{-i}}\right)w_{-i,k}\right] \nonumber \\
    w_{i,k}^*  =  & -\left(S_{w_iw_i}-S_{w_iu_i}S_{u_iu_i}^{-1}S_{u_iw_i}\right)^{-1} \left[\left(-S_{w_iu_i}S_{u_iu_i}^{-1}S_{u_i\delta_i}+S_{w_i\delta_i}\right)\delta_{i,k} \nonumber \right. \\ &\left. +\left(S_{w_iu_{-i}}-S_{w_iu_i}S_{u_iu_i}^{-1}S_{u_iu_{-i}}\right)u_{-i,k} +\left(S_{w_iw_{-i}}-S_{w_iu_i}S_{u_iu_i}^{-1}S_{u_iw_{-i}}\right)w_{-i,k}\right]
\end{align}

\subsubsection{Actor-disturber-critic neural network}
In this section, critic network $\hat{Q}_{i,k}$, actor networks $\hat{u}_{i,k}$ and disturber networks $\hat{w}_{i,k}$ are used to approximate the Q-value function $Q_{i,k}$, control policy $u_{i,k}$ and disturbance input policy $w_{i,k}$ for each agent, respectively.
The Q-function is approximated by 
\begin{align}\label{e4.2.1}
    \hat{Q}_{i,k} = z_{i,k}^\top W_{c,i}^k z_{i,k},
\end{align}
where $W_{c,i}^k$ is the critic weights.
The objective of the critic network is to minimize the square residual error
\begin{align*}
    E_{c,i,k} = \frac{1}{2}e_{c,i,k}^\top e_{c,i,k},
\end{align*}
where the temporal difference (TD) error $e_{c,i,k}$ is defined as
\begin{align}\label{e4.2.3}
    e_{c,i,k} = r_{i,k}
     + \hat{Q}_{i,k+1} - \hat{Q}_{i,k}.
\end{align}

The tuning law for the critic NN is designed as
\begin{align}\label{e4.2.4}
    W_{c,i}^{k+1} & = W_{c,i}^k - \alpha_{c,i}\frac{\partial E_{c,i,k}}{\partial e_{c,i,k}}\frac{\partial e_{c,i,k}}{\partial W_{c,i}^k}\nonumber\\
    & = W_{c,i}^k - \alpha_{c,i}e_{c,i,k}(z_{i,k+1}z_{i,k+1}^\top-z_{i,k}z_{i,k}^\top),
\end{align}
where $\alpha_{c,i}$ is the learning rate of the $i^{th}$ critic NN.

The actor NN to approximate the control policy is designed as
\begin{align}\label{e4.2.5}
    \hat{u}_{i,k} = {W_{a,i}^k}^\top \phi(\delta_{i,k}),
\end{align}
where $W_{a,i}^k \in \mathbb{R}^{m \times p} $ is the actor weights and $\phi(\cdot) \in \mathbb{R}^m$ is the basis function of actor NN. 

The approximation error of the actor is defined as
\begin{align}\label{e4.2.6}
    e_{a,i,k} = \hat{u}_{i,k} - u_{i,k}^{Q},
\end{align}
where $u_{i,k}^{Q}$ is the target control policy given in \eqref{e4.1.2} by replacing $S_i$ with $W_{c,i}^k$.

The objective function of actor can be formalized as follows to minimize the approximation error
\begin{align*}
    E_{a,i,k}=\frac{1}{2}e_{a,i,k}^\top e_{a,i,k}
\end{align*}
via updating the actor network weights
\begin{align}\label{e4.2.7}
    W_{a,i}^{k+1}  & = W_{a,i}^{k} - \alpha_{a,i} \frac{\partial E_{a,i,k}}{\partial e_{a,i,k}} \frac{\partial e_{a,i,k}}{\partial \hat{u}_{i,k}} \frac{\partial \hat{u}_{i,k}}{\partial W_{a,i}^k}\nonumber \\ &
    = W_{a,i}^{k} - \alpha_{a,i} \phi(\delta_{i,k})e_{a,i,k}^\top,
\end{align}
where $\alpha_{a,i}$ is the actor NN learning rate.

Similarly, the disturbance policy is approximated by the policy
\begin{align}\label{e4.2.8}
    \hat{w}_{i,k} = {W_{d,i}^k}^\top \varphi(\delta_{i,k}),
\end{align}
where $W_{d,i}^k \in \mathbb{R}^{o \times p} $ is the disturber weights and $\varphi(\cdot) \in \mathbb{R}^o$ is the basis function of disturber NN. 

The approximation error is defined as
\begin{align}\label{e4.2.9}
    e_{d,i,k} = \hat{w}_{i,k} - w_{i,k}^{Q},
\end{align}
where $w_{i,k}^{Q}$ is the target disturbance policy given in \eqref{e4.1.2} by replacing $S_i$ with $W_{c,i}^k$.

The square residual error is defined for the disturber as 
\begin{align*}
    E_{d,i,k}=\frac{1}{2} e_{d,i,k}^\top e_{d,i,k}.
\end{align*}
by adjusting the NN parameters
\begin{align}\label{e4.2.10}
    W_{d,i}^{k+1}  & = W_{d,i}^{k} - \alpha_{d,i} \frac{\partial E_{d,i,k}}{\partial e_{d,i,k}} \frac{\partial e_{d,i,k}}{\partial \hat{w}_{i,k}} \frac{\partial \hat{w}_{i,k}}{\partial W_{d,i}^k}\nonumber
    \\ & = W_{d,i}^{k} - \alpha_{d,i}   \varphi(\delta_{i,k})e_{d,i,k}^\top,
\end{align}
where $\alpha_{d,i}$ is the disturber network learning rate.

\begin{remark}\label{rem2}
    Note that the online training process and design of actor, disturber and critic NNs require no information of the system dynamics $A$, $B_i$ and $E_i$ of agent $i$. The critic weights should be initialized such that $S_{u_iu_i}$ and $S_{w_iw_i}$ are nonsingular. 
\end{remark}



Let $W_{c,i}^*$, $W_{a,i}^*$ and $W_{d,i}^*$ be the optimal weights for critic NN, actor NN and disturber NN. 
Define the critic NN errors $\Tilde{W}_{c,i}^{k}$, the actor NN errors $\Tilde{W}_{a,i}^{k}$ and the disturber NN errors $\Tilde{W}_{d,i}^{k}$ as
\begin{align}\label{e4.2.11}
    \Tilde{W}_{c,i}^{k} & = W_{c,i}^{k}-W_{c,i}^* \nonumber \\
    \Tilde{W}_{a,i}^{k} & = W_{a,i}^{k}-W_{a,i}^*  \\
    \Tilde{W}_{d,i}^{k} & = W_{d,i}^{k}-W_{d,i}^* \nonumber
\end{align}
Then the dynamics of the critic, actor and disturber NN approximation errors can be obtained as 
\begin{align}\label{e4.2.12}
    \Tilde{W}_{vc,i}^{k+1} & = \Tilde{W}_{vc,i}^{k}- \alpha_{c,i}\eta(z_{i,k}) e_{c,i,k}\\ \label{e4.2.13}
    \Tilde{W}_{va,i}^{k+1} & = \Tilde{W}_{va,i}^{k}- \alpha_{a,i} \phi_v(\delta_{i,k}) e_{a,i,k} \\ \label{e4.2.14}
    \Tilde{W}_{vd,i}^{k+1} & = \Tilde{W}_{vd,i}^{k}-\alpha_{d,i} \varphi_v(\delta_{i,k}) e_{d,i,k}.
\end{align}
where $W_{vc,i}^k = \textnormal{vec}(W_{c,i}^k)$, $\eta(z_{i,k}) = z_{i,k+1} \otimes z_{i,k+1} - z_{i,k} \otimes z_{i,k}$, $W_{va,i}^k = \textnormal{vec}(W_{a,i}^k)$, $\phi_v(\delta_{i,k}) = I_p \otimes \phi(\delta_{i,k})$, $W_{vd,i}^k = \textnormal{vec}(W_{d,i}^k)$ and $\varphi_v(\delta_{i,k}) = I_p \otimes \varphi(\delta_{i,k})$.

To guarantee that $W_{vc,i}^{k}$, $W_{va,i}^{k}$ and $W_{vd,i}^{k}$ converge to $W_{vc,i}^{*}$, $W_{va,i}^{*}$ and $W_{vd,i}^{*}$, the persistence of excitation condition should be satisfied.
\begin{definition} \citep*{ioannou2006adaptive}
    A signal $v_k \in \mathbb{R}^q$ is said to be persistently exciting over an interval $[k, k + T]$ if there exists $\beta > 0$, 
    \begin{align*}
    \mathop{\sum}\nolimits_{l=k}^{k+T} v_l v_l^\top \geq \beta I_q
    \end{align*}
    is satisfied for $\forall k = 0, 1, 2, \cdots$.
\end{definition}

The following assumption is made for showing the convergence of the online method.
\begin{assumption} \label{a4}
Given the NNs \eqref{e4.2.1}, \eqref{e4.2.5} and \eqref{e4.2.8}, the following conditions hold:
\begin{enumerate}[1)] 
    \item The optimal critic NN weights $W_{vc,i}^*$, actor NN weights $W_{va,i}^*$ and disturber NN weights, $W_{vd,i}^*$ are bounded by positive constants $\bar{W}_{c,i}$, $\bar{W}_{a,i}$ and $\bar{W}_{d,i}$.
    \item The activation function $\eta(z_{i,k})$ satisfies the persistence of excitation condition.
    \item The activation functions $\eta(z_{i,k})$, $\phi_v(\delta_{i,k})$ and $\varphi_v(\delta_{i,k})$ are bounded by positive constants $\bar{\eta}_i$, $\bar{\phi}_i$ and $\bar{\varphi}_i$.
    \item The target policies $u_{i,k}^{Q}$ and $w_{i,k}^{Q}$ are bounded by $\bar{u}_i$ and $\bar{w}_i$.
\end{enumerate}
\end{assumption}

\begin{remark}
    Assumption \ref{a4} is a standard assumption in adaptive dynamic programming using neural networks \cite{zhang2014online}.
    Assumption 3.2 can be satisfied by adding probing noise to control input \citep*{li2018off}.
    Assumption 3.3 can be satisfied by taking standard sigmoid, Gaussian or other NN activation functions \citep*{vamvoudakis2010online}. 
    The bounds are only used for showing the convergence of the proposed method and not needed in learning process.
\end{remark}

The convergence of the online model-free method to the optimal control policies and disturbance policies is shown by the next theorem.
\begin{thm} \label{thm3}
Let the critic NN, actor NN and disturber NN be given by \eqref{e4.2.1}, \eqref{e4.2.5} and \eqref{e4.2.8}.
Suppose Assumption \ref{a4} holds $\forall i = 1, 2, \cdots, N$.
Tune the NN weights by \eqref{e4.2.4}, \eqref{e4.2.7} and \eqref{e4.2.10}.
Then the critic weights errors $\Tilde{W}_{c,i}$, actor weights errors $\Tilde{W}_{a,i}$ and disturber weights errors $\Tilde{W}_{d,i}$ are uniformly ultimately bounded (UUB).
Moreover, $\hat{Q}_i$ converges to the approximate optimal solution of the DTHJI equation based on Q-function and $\hat{u}_i$, $\hat{w}_i$ converge to the approximate Nash solution of the cooperative graphical game.
\end{thm}

\begin{proof}
See Appendix.
\end{proof}

\subsection{\wxy{Online implementation for non-cooperative graphical game}}
\subsubsection{\wxy{Q-function-based optimal policies}}
\wxy{
To implement the online model-free PI method, the form of the optimal policies can be used as the target policies to update the actor, disturber and adversary.
Considering the quadratic form of the Q-function, let $\mathcal{S}_i$ be the kernel matrix such that
\begin{align*}
    \mathcal{Q}_{i,k}  =&\ z_{i,k}^\top \mathcal{S}_iz_{i,k} \\
    = &\ z_{i,k}^\top
    \begin{bmatrix}
        \mathcal{S}_{\delta_i \delta_i} & \mathcal{S}_{\delta_i u_i} & \mathcal{S}_{\delta_i u_{-i}} & \mathcal{S}_{\delta_i w_i} & \mathcal{S}_{\delta_i w_{-i}} \\ 
        \mathcal{S}_{u_i \delta_i} & \mathcal{S}_{u_i u_i} & \mathcal{S}_{u_i u_{-i}} & \mathcal{S}_{u_i w_i} & \mathcal{S}_{u_i w_{-i}} \\
        \mathcal{S}_{u_{-i} \delta_i} & \mathcal{S}_{u_{-i} u_i} & \mathcal{S}_{u_{-i} u_{-i}} & \mathcal{S}_{u_{-i} w_i} & \mathcal{S}_{u_{-i} w_{-i}}\\
        \mathcal{S}_{w_i \delta_i} & \mathcal{S}_{w_i u_i} & \mathcal{S}_{w_i u_{-i}} & \mathcal{S}_{w_i w_i} & \mathcal{S}_{w_i w_{-i}}\\
        \mathcal{S}_{w_{-i} \delta_i} & \mathcal{S}_{w_{-i} u_i} & \mathcal{S}_{w_{-i} u_{-i}} & \mathcal{S}_{w_{-i} w_i} & \mathcal{S}_{w_{-i} w_{-i}}
    \end{bmatrix}
        z_{i,k},
\end{align*}
where $z_{i,k} = \col(\delta_{i,k},u_{i,k},u_{-i,k},w_{i,k},w_{-i,k})$.}

\wxy{
Using the optimal condition $\frac{\partial \mathcal{Q}_{i,k}}{\partial u_{i,k}}=0$ and $\frac{\partial \mathcal{Q}_{i,k}}{\partial w_{i,k}}=0$, $u_{i,k}^*$ and $w_{i,k}^*$ can be calculated as
\begin{align} \label{non_target}
    \begin{bmatrix}
        u_{i,k}^* \\
        u_{-i,k}^* \\
        w_{i,k}^* \\
        w_{-i,k}^*
    \end{bmatrix}
    = -
    \begin{bmatrix}
        \mathcal{S}_{u_i u_i} & \mathcal{S}_{u_i u_{-i}} & \mathcal{S}_{u_i w_i} & \mathcal{S}_{u_i w_{-i}} \\
        \mathcal{S}_{u_{-i} u_i} & \mathcal{S}_{u_{-i} u_{-i}} & \mathcal{S}_{u_{-i} w_i} & \mathcal{S}_{u_{-i} w_{-i}}\\
        \mathcal{S}_{w_i u_i} & \mathcal{S}_{w_i u_{-i}} & \mathcal{S}_{w_i w_i} & \mathcal{S}_{w_i w_{-i}}\\
        \mathcal{S}_{w_{-i} u_i} & \mathcal{S}_{w_{-i} u_{-i}} & \mathcal{S}_{w_{-i} w_i} & \mathcal{S}_{w_{-i} w_{-i}}
    \end{bmatrix}^{-1}
    \begin{bmatrix}
        \mathcal{S}_{u_i \delta_i} \\
        \mathcal{S}_{u_{-i} \delta_i} \\
        \mathcal{S}_{w_i \delta_i} \\
        \mathcal{S}_{w_{-i} \delta_i}
    \end{bmatrix}
    \delta_{i,k}.
\end{align}
}
\subsubsection{\wxy{Actor-disturber-adversary-critic neural network}}
\wxy{
In this section, critic network $\hat{\mathcal{Q}}_{i,k}$, actor networks $\hat{u}_{i,k}$, disturber networks $\hat{w}_{i,k}$ and adversary networks $\hat{u}_{ij,k}, \hat{w}_{ij,k}$ are used to approximate the Q-value function $\mathcal{Q}_{i,k}$, control policy $u_{i,k}$, disturbance input policy $w_{i,k}$ and adversary input policy $u_{ij,k}, w_{ij,k}$ for each agent $i$ and its neighbor $j$, $j \in N_i$, respectively.
The Q-function is approximated by 
\begin{align}\label{non_NN_1}
    \hat{\mathcal{Q}}_{i,k} = \Breve{z}_{i,k}^\top \mathcal{W}_{c,i}^k \Breve{z}_{i,k},
\end{align}
where $\Breve{z}_{i,k} = \col(\delta_{i,k},\hat{u}_{i,k},\hat{u}_{ij,k},\hat{w}_{i,k},\hat{w}_{ij,k})$, $j \in N_i$, and $\mathcal{W}_{c,i}^k$ is the critic weights.
The objective of the critic network is to minimize the square residual error
\begin{align*}
    \mathcal{E}_{c,i,k} = \frac{1}{2}\Breve{e}_{c,i,k}^\top \Breve{e}_{c,i,k},
\end{align*}
where the temporal difference (TD) error is defined as
\begin{align*}
    \Breve{e}_{c,i,k}  &  = \Breve{r}_{i,k}
     + \hat{\mathcal{Q}}_{i,k+1} - \hat{\mathcal{Q}}_{i,k}.
\end{align*}
}

\wxy{
The tuning law for the critic NN is designed as
\begin{align}\label{non_NN_tune_1}
    \mathcal{W}_{c,i}^{k+1} & = \mathcal{W}_{c,i}^k - \Breve{\alpha}_{c,i}\frac{\partial \mathcal{E}_{c,i,k}}{\partial \Breve{e}_{c,i,k}}\frac{\partial \Breve{e}_{c,i,k}}{\partial \mathcal{W}_{c,i}^k}\nonumber\\
    & = \mathcal{W}_{c,i}^k - \Breve{\alpha}_{c,i} \Breve{e}_{c,i,k}(\Breve{z}_{i,k+1}\Breve{z}_{i,k+1}^\top-\Breve{z}_{i,k}\Breve{z}_{i,k}^\top),
\end{align}
where $\Breve{\alpha}_{c,i}$ is the learning rate of the $i^{th}$ critic NN.}

\wxy{
The actor NN to approximate the control policy is designed as
\begin{align}\label{non_NN_2}
    \hat{u}_{i,k} = {\mathcal{W}_{a,i}^k}^\top \Breve{\phi}(\delta_{i,k}),
\end{align}
where $\mathcal{W}_{a,i}^k \in \mathbb{R}^{m \times p} $ is the actor weights and $\Breve{\phi}(\cdot) \in \mathbb{R}^m$ is the basis function of actor NN. }

\wxy{
The approximation error of the actor is defined as
\begin{align*}
    \Breve{e}_{a,i,k} = \hat{u}_{i,k} - u_{i,k}^{\mathcal{Q}},
\end{align*}
where $u_{i,k}^{\mathcal{Q}}$ is the target control policy given in \eqref{non_target} by replacing $\mathcal{S}_i$ with $\mathcal{W}_{c,i}^k$.}

\wxy{
The objective function of the actor can be formalized as follows to minimize the approximation error
\begin{align*}
    \mathcal{E}_{a,i,k}=\frac{1}{2}\Breve{e}_{a,i,k}^\top \Breve{e}_{a,i,k}
\end{align*}
via updating the actor network weights
\begin{align}\label{non_NN_tune_2}
    \mathcal{W}_{a,i}^{k+1}  & = \mathcal{W}_{a,i}^{k} - \Breve{\alpha}_{a,i} \frac{\partial \mathcal{E}_{a,i,k}}{\partial \Breve{e}_{a,i,k}} \frac{\partial \Breve{e}_{a,i,k}}{\partial \hat{u}_{i,k}} \frac{\partial \hat{u}_{i,k}}{\partial \mathcal{W}_{a,i}^k}\nonumber \\ &
    = \mathcal{W}_{a,i}^{k} - \Breve{\alpha}_{a,i} \Breve{\phi}(\delta_{i,k})\Breve{e}_{a,i,k}^\top,
\end{align}
where $\Breve{\alpha}_{a,i}$ is the actor NN learning rate.}

\wxy{
Similarly, the disturbance policy is approximated by the policy
\begin{align}\label{non_NN_3}
    \hat{w}_{i,k} = {\mathcal{W}_{d,i}^k}^\top \Breve{\varphi}(\delta_{i,k}),
\end{align}
where $W_{d,i}^k \in \mathbb{R}^{o \times p} $ is the disturber weights and $\Breve{\varphi}(\cdot) \in \mathbb{R}^o$ is the basis function of disturber NN.} 

\wxy{
The approximation error is defined as
\begin{align*}
    \Breve{e}_{d,i,k} = \hat{w}_{i,k} - w_{i,k}^{\mathcal{Q}},
\end{align*}
where $w_{i,k}^{\mathcal{Q}}$ is the target disturbance policy given in \eqref{non_target} by replacing $\mathcal{S}_i$ with $\mathcal{W}_{c,i}^k$.}

\wxy{
The square residual error is defined for the disturber as 
\begin{align*}
    \mathcal{E}_{d,i,k}=\frac{1}{2} \Breve{e}_{d,i,k}^\top \Breve{e}_{d,i,k}.
\end{align*}
by adjusting the NN parameters
\begin{align}\label{non_NN_tune_3}
    \mathcal{W}_{d,i}^{k+1}  & = \mathcal{W}_{d,i}^{k} - \Breve{\alpha}_{d,i} \frac{\partial E_{d,i,k}}{\partial \Breve{e}_{d,i,k}} \frac{\partial \Breve{e}_{d,i,k}}{\partial \hat{w}_{i,k}} \frac{\partial \hat{w}_{i,k}}{\partial \mathcal{W}_{d,i}^k}\nonumber
    \\ & = \mathcal{W}_{d,i}^{k} - \Breve{\alpha}_{d,i}   \Breve{\varphi}(\delta_{i,k})\Breve{e}_{d,i,k}^\top,
\end{align}
where $\Breve{\alpha}_{di}$ is the disturber network learning rate.}

\wxy{
The worst-case input and disturbance policy of agent $i$'s neighborhood are approximated by 
\begin{align}
    \hat{u}_{ij,k} = {\mathcal{W}_{a, ij}^k}^\top \Breve{\sigma}(\delta_{i,k}) \label{non_NN_4}, \\
    \hat{w}_{ij,k} = {\mathcal{W}_{d, ij}^k}^\top \Breve{\psi}(\delta_{i,k}) \label{non_NN_5}
\end{align}
for all $j \in N_i$, where $\mathcal{W}_{a, ij} \in \mathbb{R}^{m \times p}$ and $\mathcal{W}_{d, ij} \in \mathbb{R}^{o \times p}$ are the adversary NN weights of neighbour $j$ with respect to agent $i$, $\Breve{\sigma}(\cdot) \in \mathbb{R}^m$ and  $\Breve{\psi}(\cdot) \in \mathbb{R}^o$ are the basis functions.}  

\wxy{
By defining approximation errors
\begin{align*}
    \Breve{e}_{a, ij, k} = \hat{u}_{ij,k} - u_{ij,k}^{\mathcal{Q}}, \\
    \Breve{e}_{d, ij, k} = \hat{w}_{ij,k} - d_{ij,k}^{\mathcal{Q}},
\end{align*}
the square residual error for each adversarial input and disturbance NN of agent $i$ are given as
\begin{align*}
    \mathcal{E}_{a, ij, k}=\frac{1}{2} \Breve{e}_{a, ij, k}^\top \Breve{e}_{a, ij, k}, \\
    \mathcal{E}_{d, ij, k}=\frac{1}{2} \Breve{e}_{d, ij, k}^\top \Breve{e}_{d, ij, k},
\end{align*}
where $u_{ij,k}^{\mathcal{Q}}$ and $d_{ij,k}^{\mathcal{Q}}$ are the target adversary policies given in \eqref{non_target} by replacing $\mathcal{S}_i$ with $\mathcal{W}_{c,i}^k$.
To minimize $\mathcal{E}_{a, ij, k}$ and $\mathcal{E}_{d, ij, k}$, the update law for the adversary NNs are designed as 
\begin{align}
    \mathcal{W}_{a,ij}^{k+1}  & = \mathcal{W}_{a,ij}^{k} - \Breve{\alpha}_{a,ij} \frac{\partial E_{a,ij,k}}{\partial \Breve{e}_{a,ij,k}} \frac{\partial \Breve{e}_{a,ij,k}}{\partial \hat{u}_{ij,k}} \frac{\partial \hat{u}_{ij,k}}{\partial \mathcal{W}_{a,ij}^k}\nonumber \\ 
    & = \mathcal{W}_{a,ij}^{k} - \Breve{\alpha}_{a,ij} \Breve{\sigma}  (\delta_{i,k})\Breve{e}_{a,ij,k}^\top, \label{non_NN_tune_4} \\
    \mathcal{W}_{d,ij}^{k+1}  & = \mathcal{W}_{d,ij}^{k} - \Breve{\alpha}_{d,ij} \frac{\partial E_{d,ij,k}}{\partial \Breve{e}_{d,ij,k}} \frac{\partial \Breve{e}_{d,ij,k}}{\partial \hat{w}_{ij,k}} \frac{\partial \hat{w}_{ij,k}}{\partial \mathcal{W}_{d,ij}^k}\nonumber \\ 
    & = \mathcal{W}_{d,ij}^{k} - \Breve{\alpha}_{d,ij}   \Breve{\psi}(\delta_{i,k})\Breve{e}_{d,ij,k}^\top \label{non_NN_tune_5},
\end{align}
where $\Breve{\alpha}_{a,ij}$ and $\Breve{\alpha}_{d,ij}$ are the adversary input and disturber network learning rate, respectively.}

\wxy{
Let $ \mathcal{W}_{c,i}^*$, $ \mathcal{W}_{a,i}^*$, $ \mathcal{W}_{d,i}^*$, $ \mathcal{W}_{a,ij}^*$ and $ \mathcal{W}_{d,ij}^*$ be the optimal weights for critic NN, actor NN, disturber NN and adversary NNs.
Define the critic NN errors $\Tilde{ \mathcal{W}}_{c,i}^{k}$, the actor NN errors $\Tilde{ \mathcal{W}}_{a,i}^{k}$, the disturber NN errors $\Tilde{ \mathcal{W}}_{d,i}^{k}$ and the adeversary NN errors $\Tilde{ \mathcal{W}}_{a,ij}^{k}, \Tilde{ \mathcal{W}}_{d,ij}^{k}$ as
$\Tilde{\mathcal{W}}_{c,i}^{k} =\mathcal{W}_{c,i}^{k}-\mathcal{W}_{c,i}^*,
    \Tilde{\mathcal{W}}_{a,i}^{k} =\mathcal{W}_{a,i}^{k}-\mathcal{W}_{a,i}^*, 
    \Tilde{\mathcal{W}}_{d,i}^{k} =\mathcal{W}_{d,i}^{k}-\mathcal{W}_{d,i}^*, 
    \Tilde{\mathcal{W}}_{a,i}^{k} =\mathcal{W}_{a,ij}^{k}-\mathcal{W}_{a,ij}^*,  
    \Tilde{\mathcal{W}}_{d,i}^{k} =\mathcal{W}_{d,ij}^{k}-\mathcal{W}_{d,ij}^*$.
Then the dynamics of the critic, actor and disturber NN approximation errors can be obtained as 
\begin{align*}\label{e4.2.12}
    \Tilde{\mathcal{W}}_{vc,i}^{k+1} & = \Tilde{\mathcal{W}}_{vc,i}^{k}- \Breve{\alpha}_{c,i}\Breve{\eta}(\Breve{z}_{i,k}) \Breve{e}_{c,i,k},\\ 
    \Tilde{\mathcal{W}}_{va,i}^{k+1} & = \Tilde{\mathcal{W}}_{va,i}^{k}- \Breve{\alpha}_{a,i} \Breve{\phi}_v(\delta_{i,k}) \Breve{e}_{a,i,k}, \\
    \Tilde{\mathcal{W}}_{vd,i}^{k+1} & = \Tilde{\mathcal{W}}_{vd,i}^{k}- \Breve{\alpha}_{d,i} \Breve{\varphi}_v(\delta_{i,k}) \Breve{e}_{d,i,k}, \\
    \Tilde{\mathcal{W}}_{va,ij}^{k+1} & = \Tilde{\mathcal{W}}_{va,ij}^{k}- \Breve{\alpha}_{a,ij} \Breve{\sigma}_v(\delta_{i,k}) \Breve{e}_{a,i,k}, \\
    \Tilde{\mathcal{W}}_{vd,ij}^{k+1} & = \Tilde{\mathcal{W}}_{vd,ij}^{k}- \Breve{\alpha}_{d,ij} \Breve{\psi}_v(\delta_{i,k}) \Breve{e}_{d,ij,k}.
\end{align*}
where $\Tilde{\mathcal{W}}_{vc,i}^k = \textnormal{vec}(\Tilde{\mathcal{W}}_{c,i}^k)$, $\Breve{\eta}(z_{i,k}) = \Breve{z}_{i,k+1} \otimes \Breve{z}_{i,k+1} - \Breve{z}_{i,k} \otimes \Breve{z}_{i,k}$, $\Tilde{\mathcal{W}}_{va,i}^k = \textnormal{vec}(W_{a,i}^k)$, $\Breve{\phi}_v(\delta_{i,k}) = I_p \otimes \Breve{\phi}(\delta_{i,k})$, $\Tilde{\mathcal{W}}_{vd,i}^k = \textnormal{vec}(W_{d,i}^k)$, $\Breve{\varphi}_v(\delta_{i,k}) = I_p \otimes \Breve{\varphi}(\delta_{i,k})$, $\Tilde{\mathcal{W}}_{va,ij}^k = \textnormal{vec}(W_{a,ij}^k)$, $\Breve{\sigma}_v(\delta_{i,k}) = I_p \otimes \Breve{\sigma}(\delta_{i,k})$, $\Tilde{\mathcal{W}}_{vd,ij}^k = \textnormal{vec}(W_{d,ij}^k)$ and $\Breve{\psi}_v(\delta_{i,k}) = I_p \otimes \Breve{\psi}(\delta_{i,k})$.
}

\wxy{
The following assumption is made for showing the convergence of the online method.
\begin{assumption} \label{a5}
Given the NNs \eqref{non_NN_1}, \eqref{non_NN_2}, \eqref{non_NN_3}, \eqref{non_NN_4} and \eqref{non_NN_5}, the following conditions hold:
\begin{itemize} 
    \item[\textbf{1)}] The optimal critic NN weights $\mathcal{W}_{c,i}^*$, actor NN weights $\mathcal{W}_{a,i}^*$, disturber NN weights $\mathcal{W}_{d,i}^*$ and adversary NN weights $\mathcal{W}_{a,ij}^*, \mathcal{W}_{d,ij}^*$ are bounded by positive constants $\bar{\mathcal{W}}_{c,i}$, $\bar{\mathcal{W}}_{a,i}$, $\bar{\mathcal{W}}_{d,i}$, $\bar{\mathcal{W}}_{a,ij}$ and $\bar{\mathcal{W}}_{d,ij}$.
    \item[\textbf{2)}] The activation function $\Breve{\eta}(z_{i,k})$ satisfies the persistence of excitation condition.
    \item[\textbf{3)}] The activation functions $\Breve{\eta}(z_{i,k})$, $\Breve{\phi}_v(\delta_{i,k})$, $\Breve{\varphi}_v(\delta_{ik})$, $\Breve{\sigma}_v(\delta_{i,k})$ and $\Breve{\psi}_v(\delta_{ik})$ are bounded by positive constants $\bar{\Breve{\eta}}_i$, $\bar{\Breve{\phi}}_i$, $\bar{\Breve{\varphi}}_i$, $\bar{\Breve{\sigma}}_i$ and $\bar{\Breve{\psi}}_i$.
    \item[\textbf{4)}] The target policies $\Breve{u}_{i,k}^{\mathcal{Q}}$, $\Breve{w}_{i,k}^{\mathcal{Q}}$, $\Breve{u}_{ij,k}^{\mathcal{Q}}$ and $\Breve{w}_{ij,k}^{\mathcal{Q}}$ are bounded by $\bar{\Breve{u}}_i$, $\bar{\Breve{w}}_i$, $\bar{\Breve{u}}_{ij}$ and $\bar{\Breve{w}}_{ij}$ for all $j \in N_i$.
\end{itemize}
\end{assumption}
}

\wxy{
The convergence of the online Q-learning method to the optimal control policies, the worst-case disturbance policies and the worst-case adversary policies is shown by the next theorem.
\begin{thm} \label{thm4}
Let the critic NN, actor NN and disturber NN be given by \eqref{non_NN_1}, \eqref{non_NN_2}, \eqref{non_NN_3}, \eqref{non_NN_4} and \eqref{non_NN_5}.
Suppose Assumption \ref{a5} holds $\forall i = 1, 2, \cdots, N$.
Tune the NN weights by \eqref{non_NN_tune_1}, \eqref{non_NN_tune_2}, \eqref{non_NN_tune_3}, \eqref{non_NN_tune_4} and \eqref{non_NN_tune_5}.
Then the critic weights errors $\Tilde{W}_{c,i}$, actor weights errors $\Tilde{W}_{a,i}$, disturber weights errors $\Tilde{W}_{d,i}$, and adversary weights errors $\Tilde{ \mathcal{W}}_{a,ij}, \Tilde{ \mathcal{W}}_{d,ij}$ are uniformly ultimately bounded (UUB), for $j = 1, 2, \cdots, N$.
Moreover, $\hat{Q}_i$ converges to the approximate optimal solution of the DTHJ equation based on the Q-function and $\hat{u}_i$ converges to the approximate minmax solution of the non-cooperative graphical game.
\end{thm}
}

\begin{proof}
\wxy{The proof is similar to the proof of Theorem \ref{thm3} and thus is omitted. }
\end{proof}

\section{Simulation results}\label{Sec Sim}

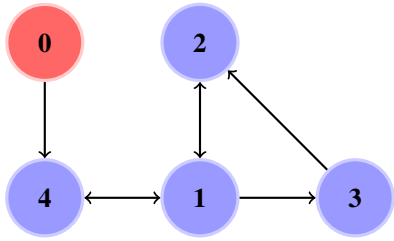
\begin{figure}
\centering
\begin{tikzpicture}[transform shape,scale=1]
    \centering%
    \node (0) [circle, draw=red!20, fill=red!60, very thick, minimum size=10mm] {\textbf{0}};
     \node (4) [circle, below=of 0, draw=blue!20, fill=blue!40, very thick, minimum size=10mm] {\textbf{4}};
    \node (1) [circle, right=of 4, draw=blue!20, fill=blue!40, very thick, minimum size=10mm] {\textbf{1}};
    \node (2) [circle, above=of 1, draw=blue!20, fill=blue!40, very thick, minimum size=10mm] {\textbf{2}};
    \node (3) [circle, right=of 1, draw=blue!20, fill=blue!40, very thick, minimum size=10mm] {\textbf{3}};
       \draw[thick,->] (0)--(4);
      \draw[thick,<->, left] (4) edge (1);
      \draw[thick,->, left] (1) edge (3);
        \draw[thick,->, left] (3) edge (2);
        \draw[thick,<->, left] (1) edge (2);
\end{tikzpicture}
 \caption{Topology structure of MAS}
    \label{fig1}
\end{figure}


\begin{figure}[htbp]
    \centering
    \subfigure[States]
        {\includegraphics[width=0.49\textwidth]{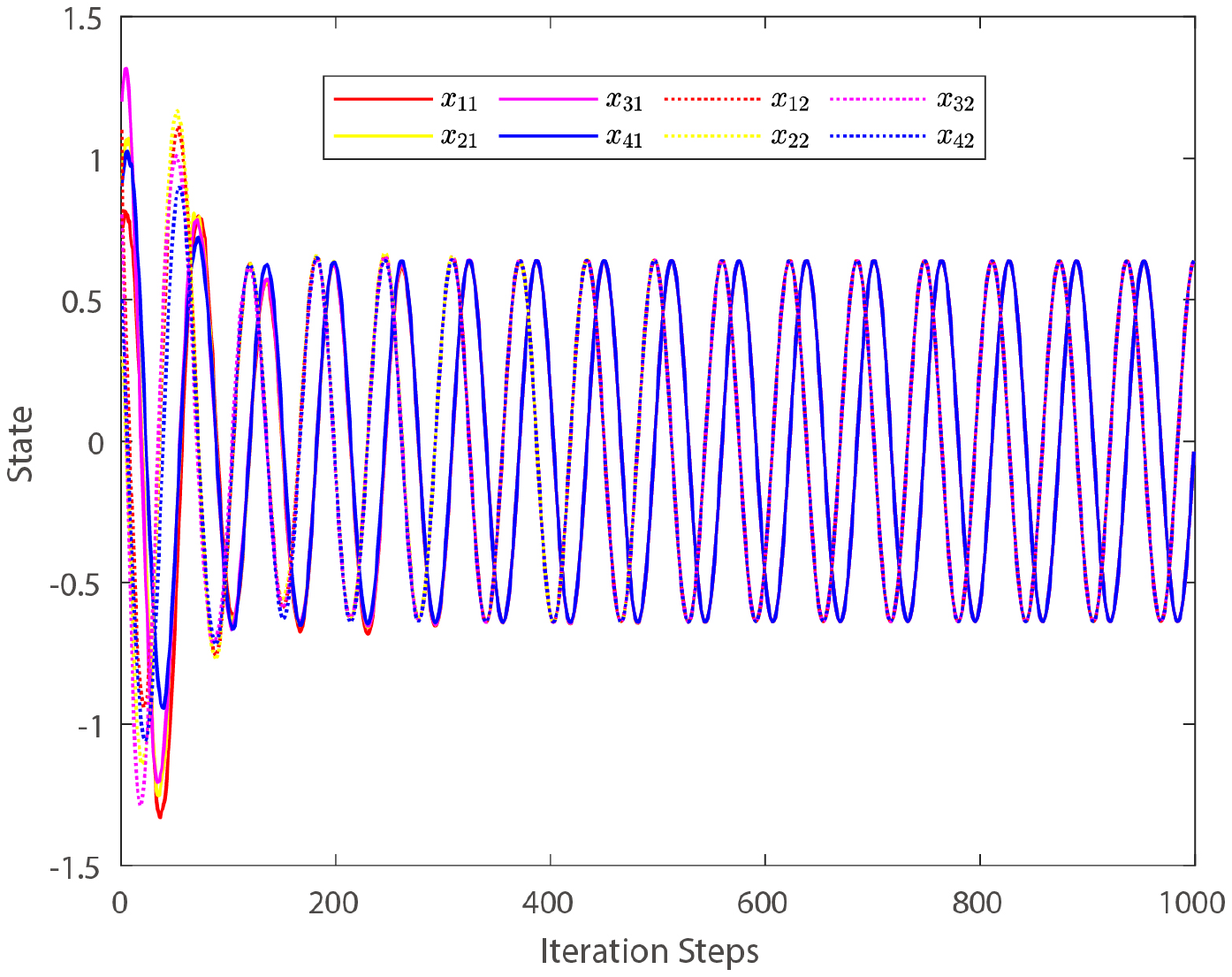}}
    \subfigure[Synchronization errors]
        {\includegraphics[width=0.49\textwidth]{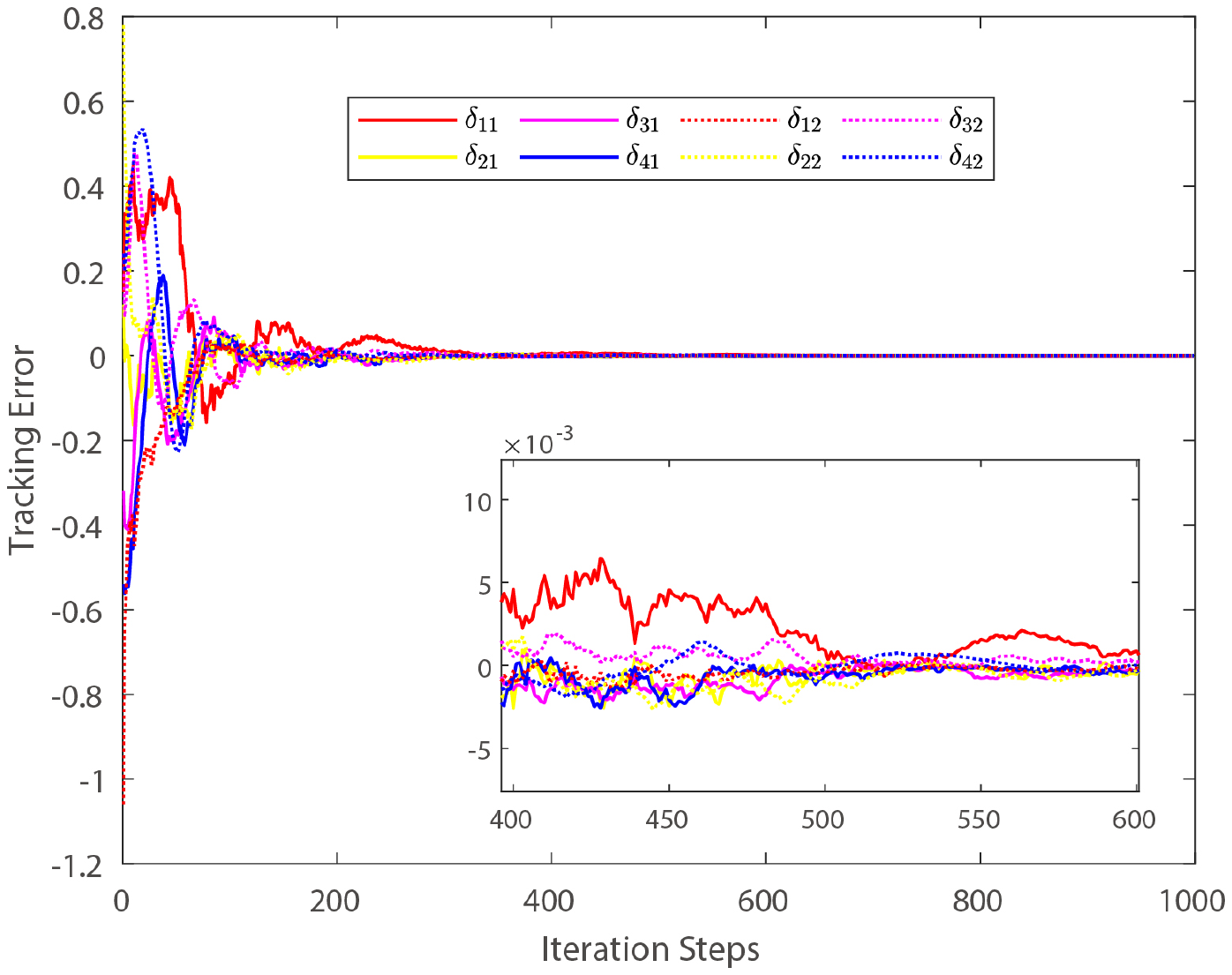}}
    \caption{Profiles of states and synchronization errors \wxy{in the cooperative graphical game}}
    \label{fig3}
\end{figure}

\begin{figure}[htbp]
    \centering
    \subfigure[Critic weights]
        {\includegraphics[width=0.49\textwidth]{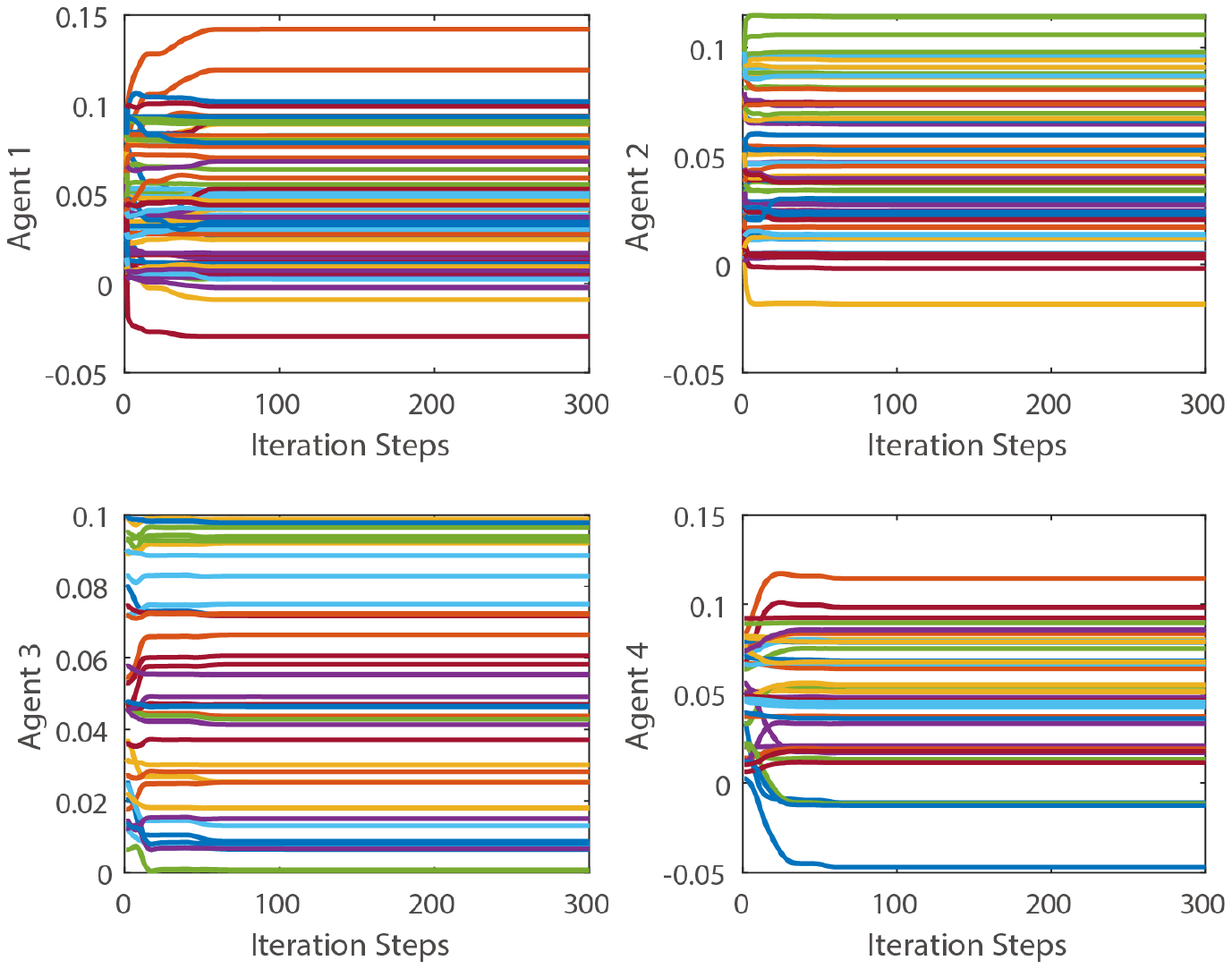}}
    \subfigure[Actor and disturber weights]
        {\includegraphics[width=0.49\textwidth]{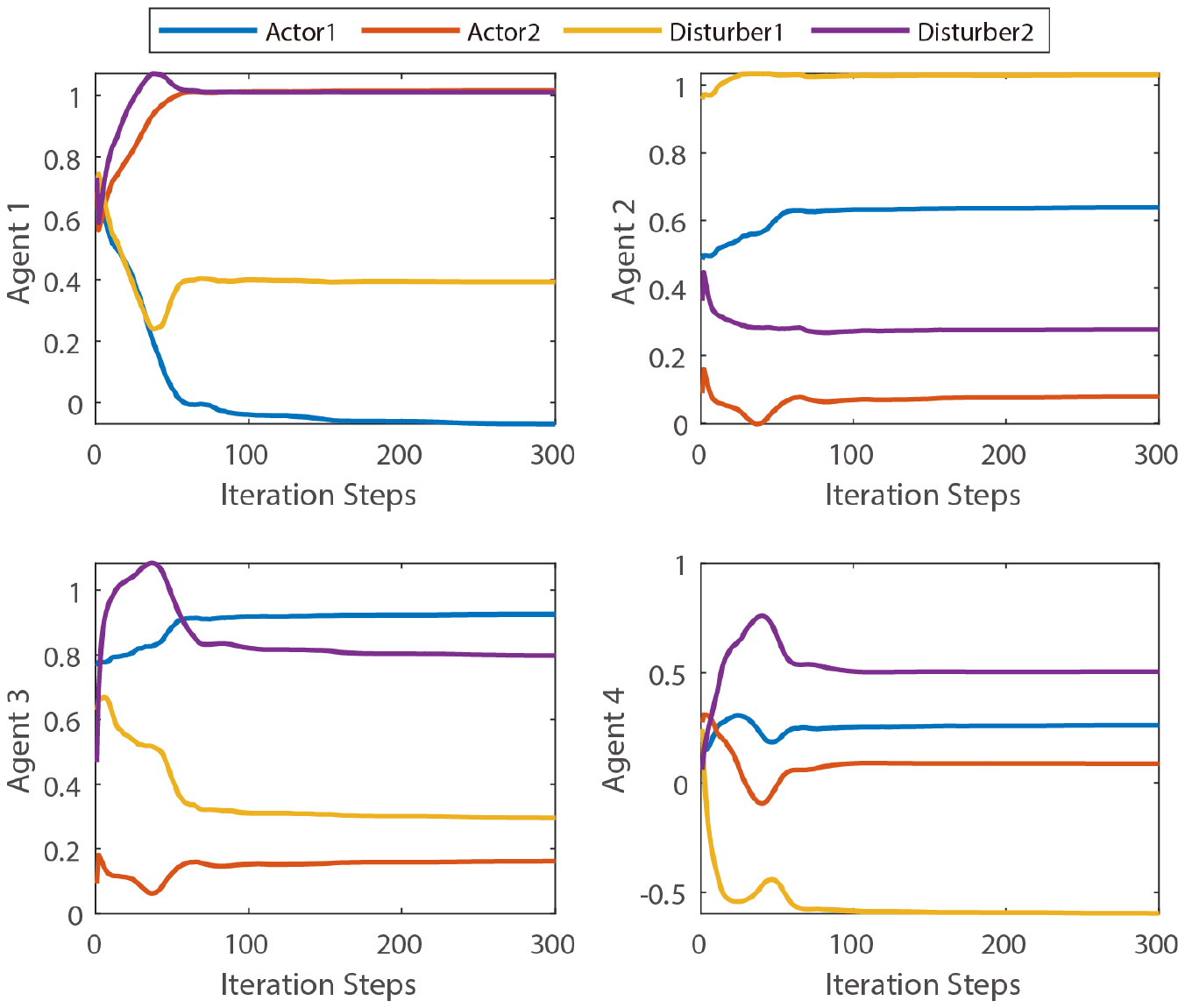}}
    \caption{Neural network weights update process \wxy{in the cooperative graphical game}}
    \label{fig4}

    \centering
    \includegraphics[width=0.5\textwidth]{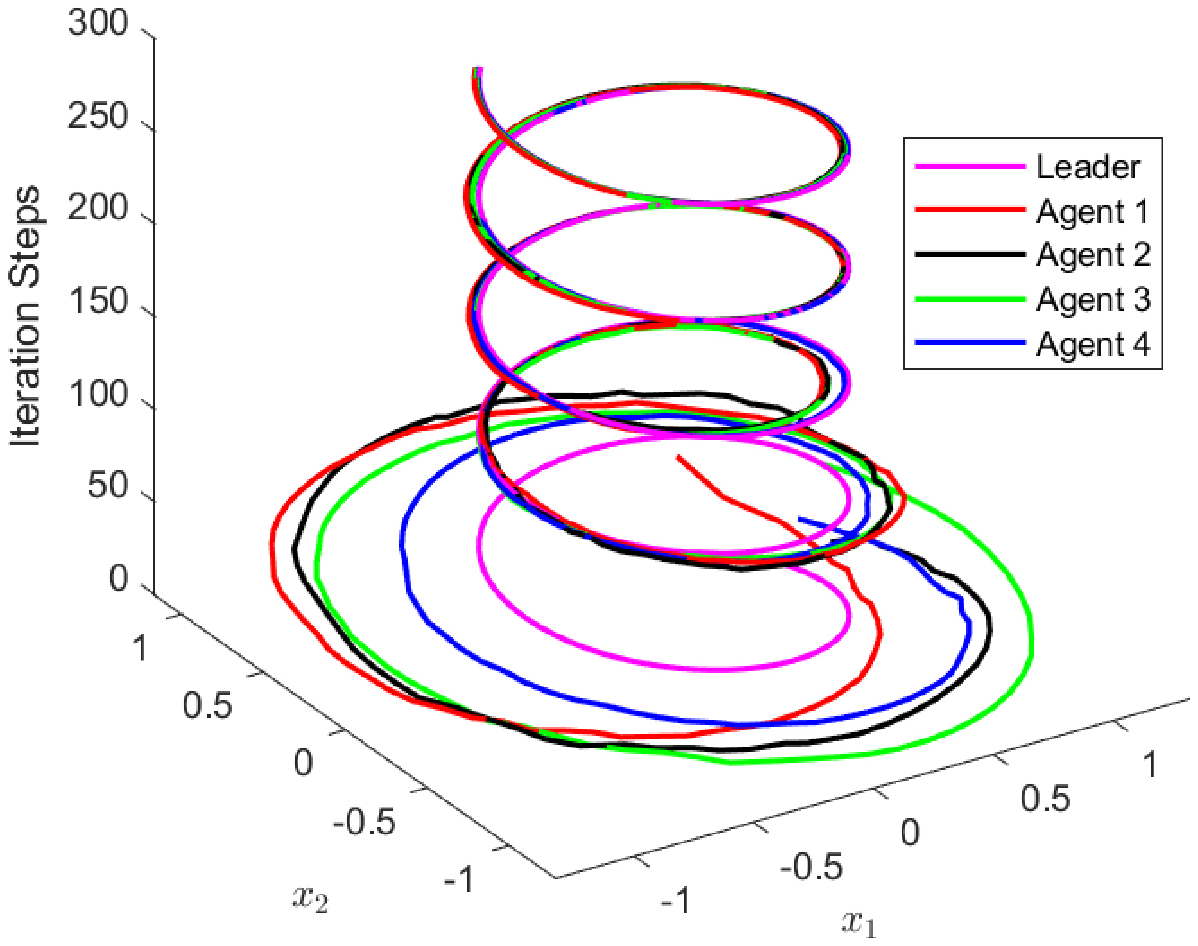}    
    \caption{3-D phase plane \wxy{in the cooperative graphical game}}  \label{fig7}  
    
\end{figure}

\begin{figure}[htbp]
    \centering
    \includegraphics[width=0.5\textwidth]{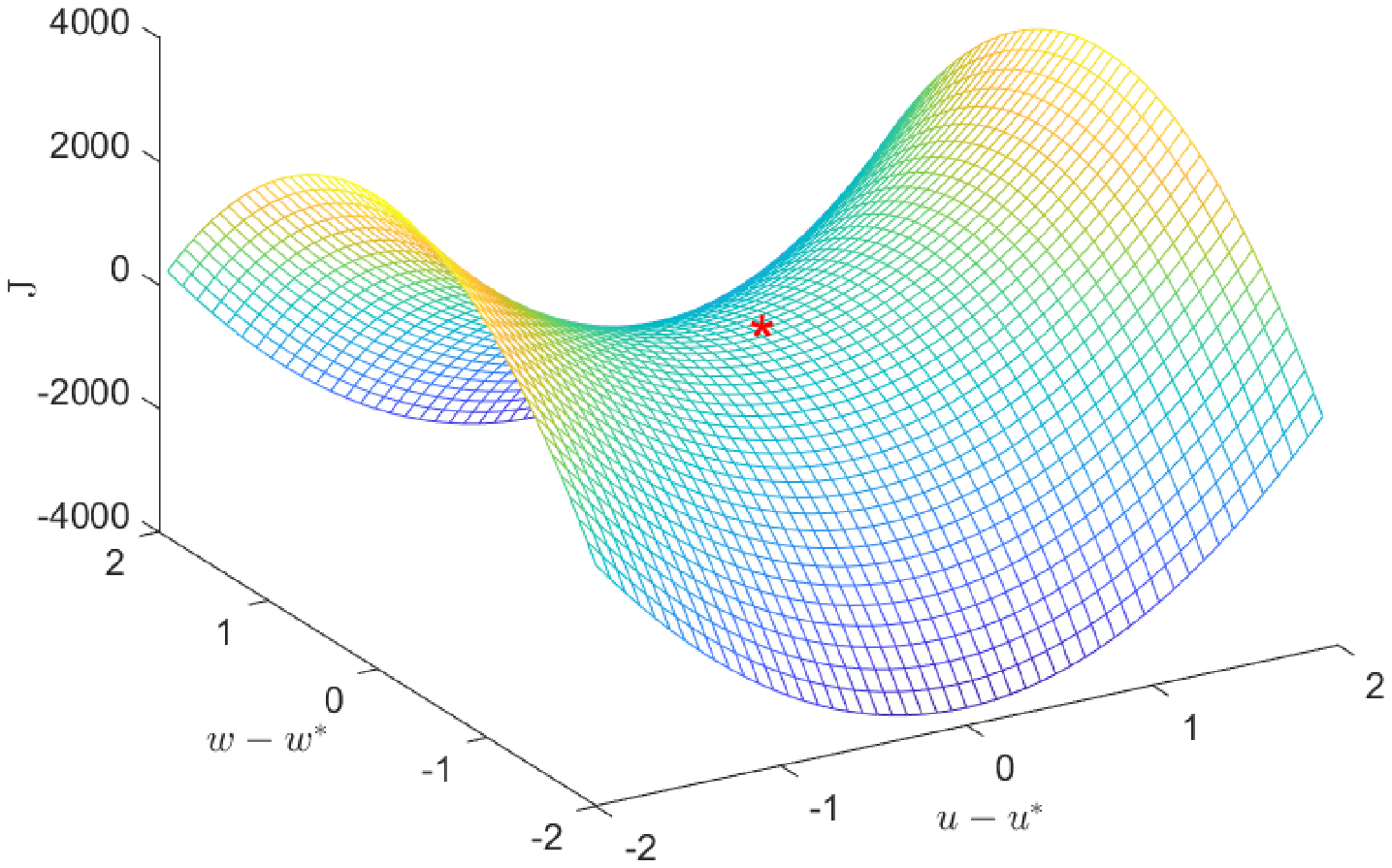}    
    \caption{Nash equilibrium}  \label{fig8}  
\end{figure}

\wxy{The system dynamics in this example are taken from Reference \zcite{abouheaf2014multi}. }
A multi-agent system with four followers and one leader under directed topology is given in Fig.\ref{fig1}.
The node labelled 0 is the leader, and the other nodes are the followers.
The edge weights and the pinning gains are given by 
\begin{align*}
e_{12} =0.8, e_{14} =0.7,e_{23} =0.6, e_{21} =0.6, e_{31} =0.8, \\
e_{41} =0.4, g_1 = g_2 = g_3 = 0, g_4 =1, e_{ij} = 0,j \not\in N_i.
\end{align*}
The dynamics for each agent are set to
\begin{equation*}
A = 
\begin{bmatrix}
0.995 & 0.09983 \\
-0.09983 & 0.995
\end{bmatrix}, B_1 = 
\begin{bmatrix}
    0.2047 \\ 0.08984
\end{bmatrix},
B_2 = 
\begin{bmatrix}
    0.2147 \\ 0.2895
\end{bmatrix},B_3 = 
\begin{bmatrix}
    0.2097 \\ 0.1897
\end{bmatrix},
B_4 = 
\begin{bmatrix}
    0.2 \\ 0.1
\end{bmatrix}.
\end{equation*}
The disturbance attenuation is set to \wxy{$\beta = \Breve{\beta} = 1$}, and
\begin{align*}
E_1 = 
\begin{bmatrix}
    0.21 \\ 0.0984
\end{bmatrix},
E_2 = 
\begin{bmatrix}
    0.32 \\ 0.084
\end{bmatrix},
E_3 = 
\begin{bmatrix}
    0.14 \\ 0.072
\end{bmatrix},
E_4 = 
\begin{bmatrix}
    0.16 \\ 0.024
\end{bmatrix}.
\end{align*}

The weighting matrices in the cost function are chosen as
\begin{align*}
&Q_{11} = I_{2 \times 2}, Q_{22} = I_{2 \times 2}, Q_{33} = I_{2 \times 2}, Q_{44} = I_{2 \times 2}, \\
   & R_{11} =1, R_{22} =1, R_{33} = 1, R_{44} = 1, R_{12} = 1,\\
    &  R_{14}  =1,R_{21} =1, R_{23} =1, R_{31} =1, R_{41} = 1, \\
    &T_{11} = 1,T_{22} = 1,T_{33} = 1, T_{12} =1,  T_{14}  =1, \\
   &T_{21} =1, T_{23} =1,  T_{31} =1, T_{41} = 1,T_{44}= 1,
\end{align*}
and $R_{ij}, T_{ij} = 0$ for $j \not\in N_i$. 
The initial states of the leader are set to be $x_0(0) = \col(0.4,0.5)$, and the initial states of the followers are chosen as $x_1(0)=\col(0.8,1.1)$, $x_2(0)=\col(0.9,0.3)$, $x_3(0)=\col(1.2,0.8)$ and $x_4(0)=\col(0.9,0.5)$.

\wxy{1) Cooperative graphical game:}
The learning rates of the critic network, the actor network and the disturber network \wxy{in the cooperative graphical game} are selected as 
\begin{equation*}
    \alpha_{c,i} = 0.1, \alpha_{a,i} = 0.1,\alpha_{d,i} = 0.1, \quad i = 1, 2, \dots, N.
\end{equation*}
The basis function for both the actor NN and the disturber NN of agent $i$ is selected as $\delta_{i,k}$.

Figure \ref{fig3} shows the states and synchronization errors evolution of all follower nodes.
Figure \ref{fig4} shows the NN weights of each agent, including critic weights, actor weights and disturber weights.
Figure \ref{fig7} shows the 3-D phase plane of agents 1, 2, 3, 4 and leader 0.
Figure \ref{fig8} shows that the $(\hat{u}_i,\hat{w}_i)$ is the saddle point of the graphical zero-sum game.
Obviously, synchronization of the MAS is achieved under the $L_2$ bounded disturbance and the Nash equilibrium seeking problem is solved.

\begin{figure}[htbp]
    \centering
    \subfigure[States]
        {\includegraphics[width=0.49\textwidth]{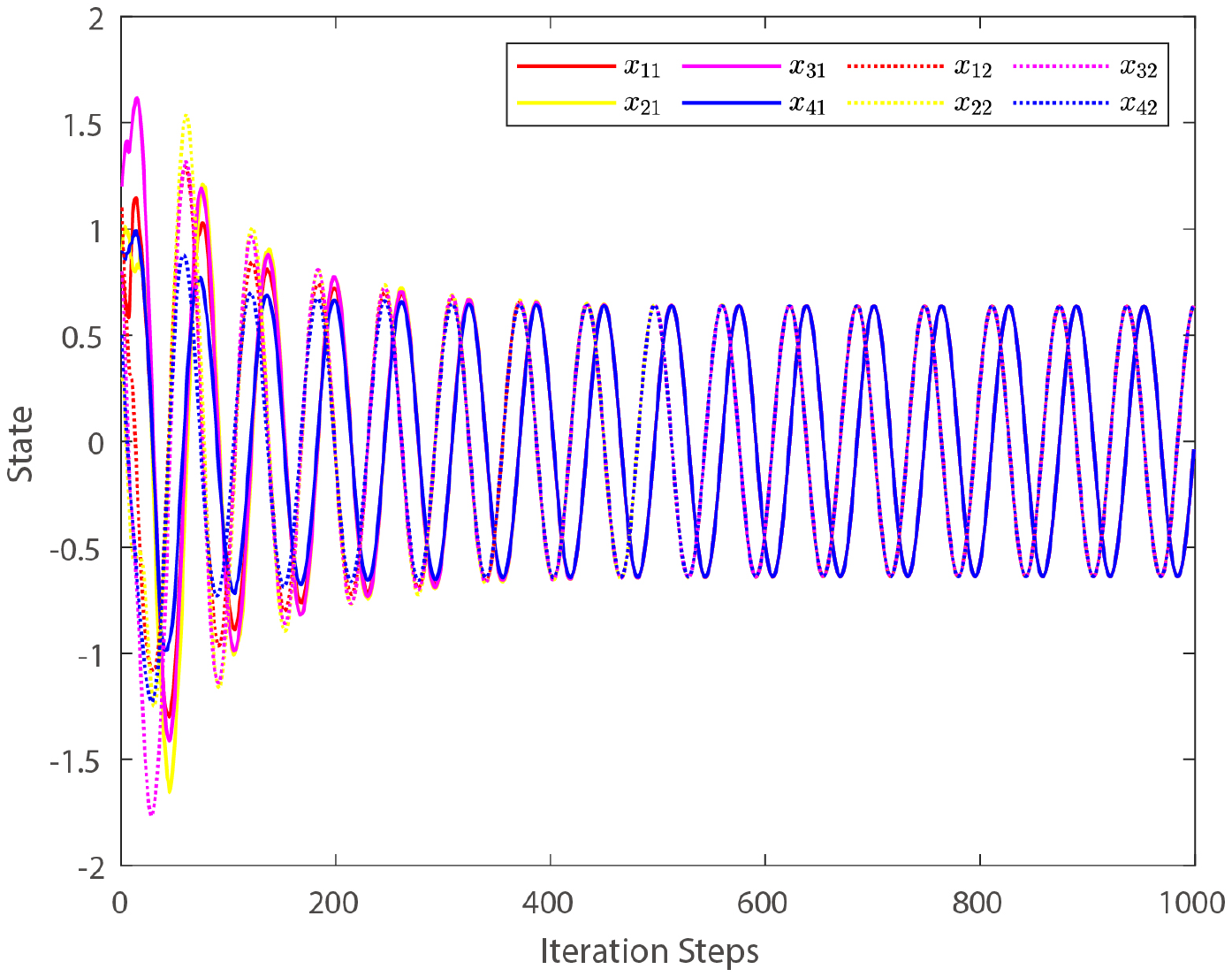}}
    \subfigure[Synchronization errors]
        {\includegraphics[width=0.49\textwidth]{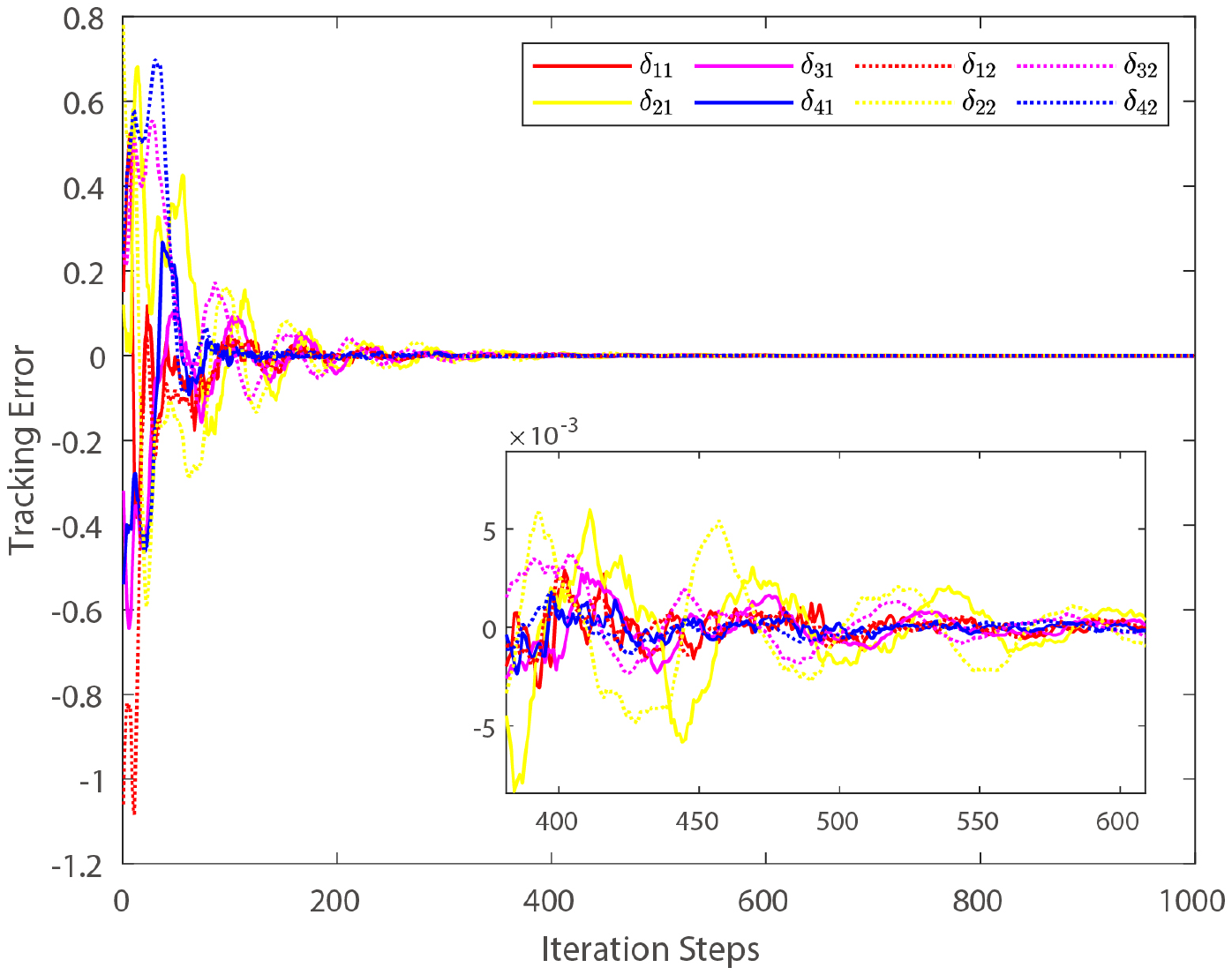}}
    \caption{\wxy{Profiles of states and synchronization errors in the non-cooperative graphical game}}
    \label{fig9}
\end{figure}

\begin{figure}[htbp]
    \centering
    \subfigure[Critic weights]
        {\includegraphics[width=0.49\textwidth]{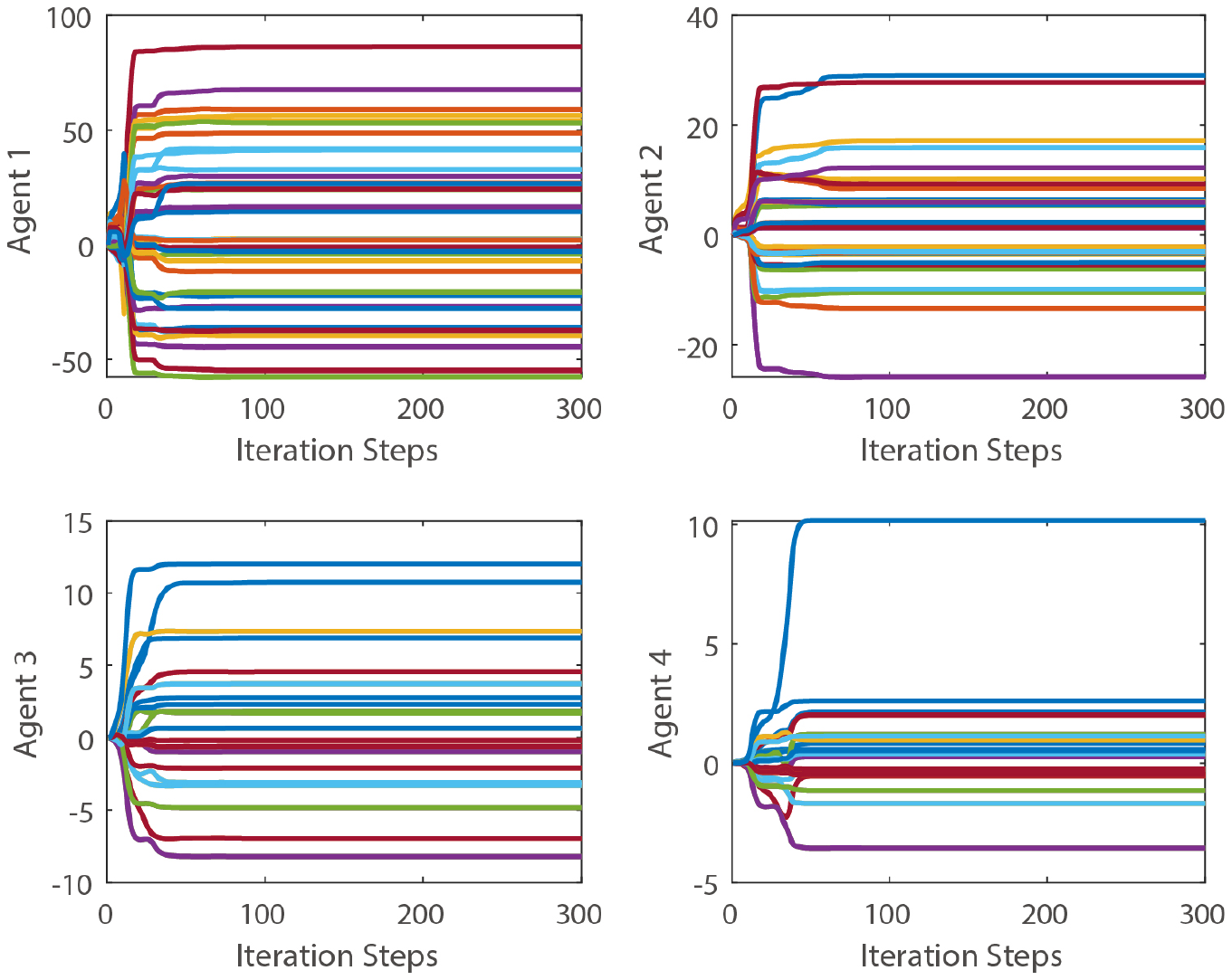}}
    \subfigure[Actor and disturber weights]
        {\includegraphics[width=0.49\textwidth]{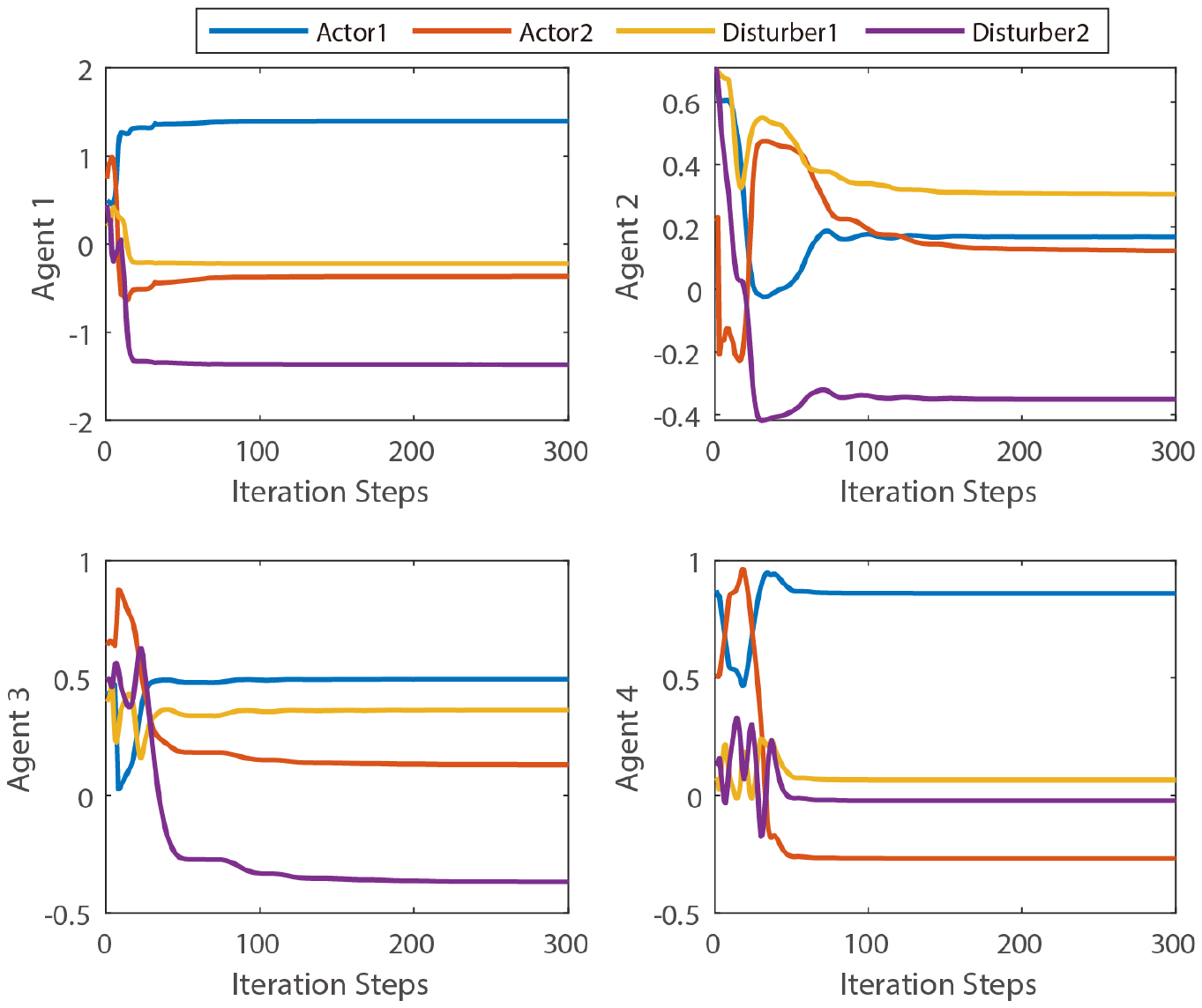}}
    \subfigure[Adversary weights]
        {\includegraphics[width=0.49\textwidth]{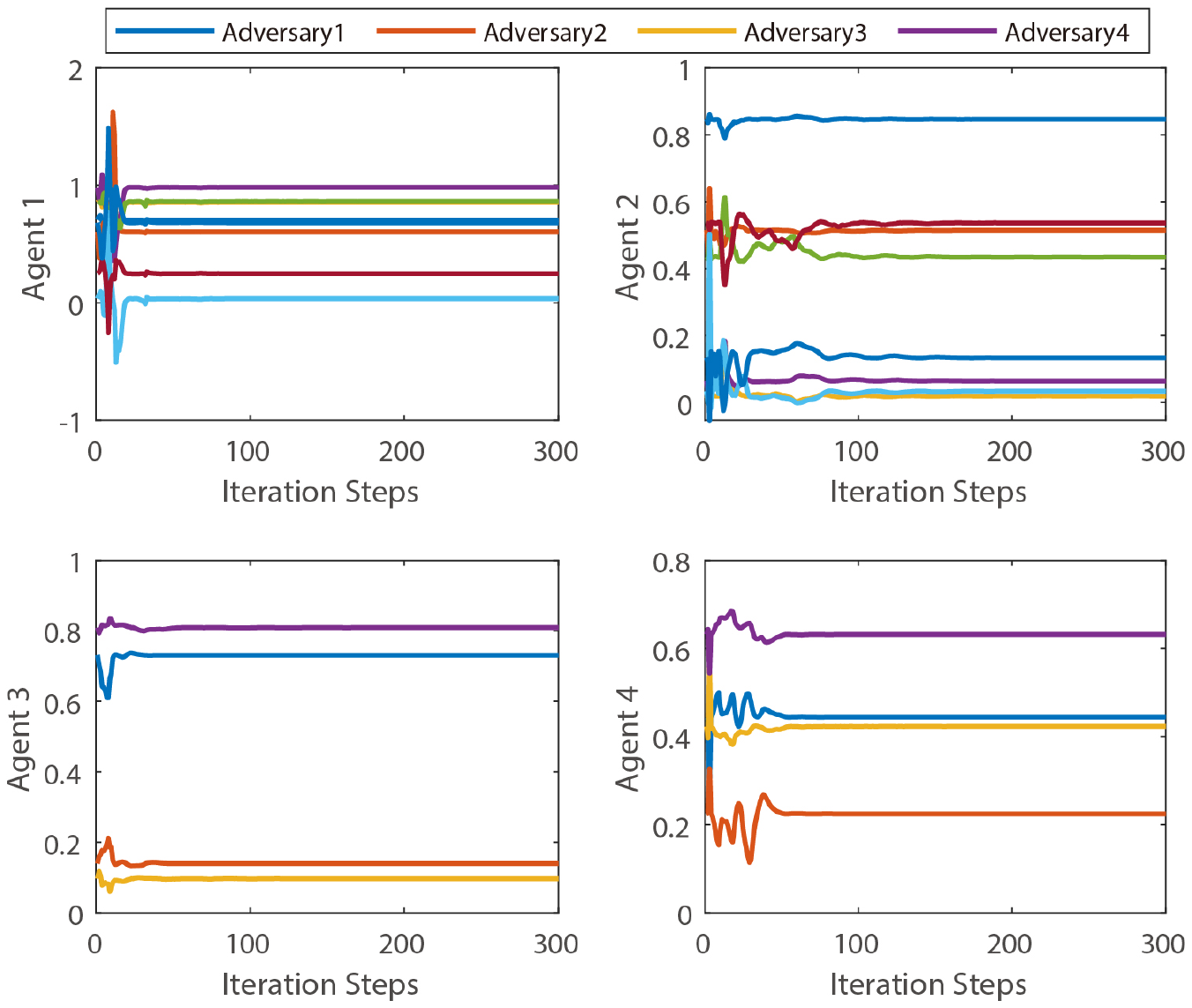}}
    \caption{\wxy{Neural network weights update process in the non-cooperative graphical game}}
    \label{fig10}

    \centering
    \includegraphics[width=0.5\textwidth]{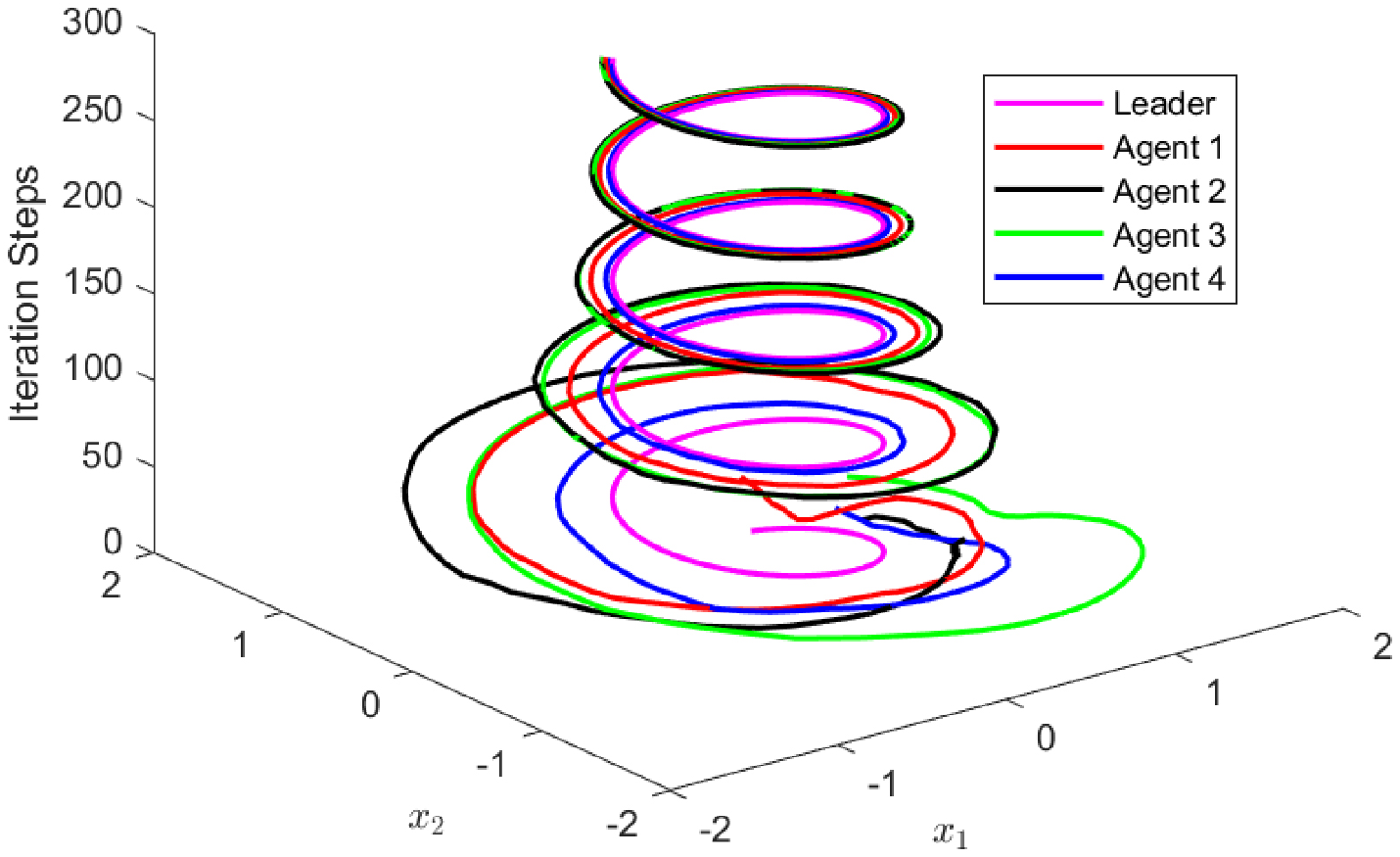}    
    \caption{\wxy{3-D phase plane in the non-cooperative graphical game}}  \label{fig11}  
\end{figure}

\wxy{2) Non-cooperative graphical game:
The learning rates of the critic network, the actor network, the disturber network and the adversary networks in the non-cooperative graphical game are selected as 
\begin{equation*}
    \Breve{\alpha}_{c,i} = 0.1,\quad\Breve{\alpha}_{a,i} = 0.1,\quad \Breve{\alpha}_{d,i} = 0.1, \quad\Breve{\alpha}_{a,ij} = 0.05,\quad\Breve{\alpha}_{d,ij} = 0.05, \quad i = 1, 2, \dots, N, \quad j \in N_i
\end{equation*}
The basis function for the actor NN, the disturber NN and the adversary NNs of agent $i$ is selected as $\delta_{i,k}$.}

\wxy{
Figure \ref{fig9} shows the states and synchronization errors evolution of all follower nodes.
Figure \ref{fig10} shows the NN weights of each agent, including critic weights, actor weights, disturber weights and adversary weights.
Figure \ref{fig11} shows the 3-D phase plane of followers and leader 0.
Synchronization of the MAS is achieved in the presence of the $L_2$ bounded disturbance.}

\section{Conclusions} \label{Sec_conclusion}
\wxy{This article studied both cooperative and non-cooperative graphical games  with disturbance rejection in discrete-time MAS.
To seek the Nash equilibrium in the cooperative graphical game and the distributed minmax solution in the non-cooperative graphical game without any knowledge of agents dynamics, Q-function based PI methods were proposed.
To implement these model-free methods, an actor-disturber-critic framework was employed to approximate the Q-functions, control policies and disturbance policies for the cooperative graphical game.
For the non-cooperative graphical game, adversary NNs were incorporated to approximate the adversary policies.
The convergence of the model-free policy iteration to both the approximate Nash equilibrium solution and the fully distributed solution were rigorously analyzed.
Simulation results showed the effectiveness and feasibility of the proposed methods.}

\section{Appendix} \label{App}
\textbf{\emph{Proof for Theorem \ref{thm3}:}}
A Lyapunov function candidate is selected as
\begin{align}\label{e4.3.2}
    L_i^k  = &\ L_{c,i}^k+L_{a,i}^k+L_{d,i}^k \nonumber \\
           = &\ \frac{1}{\alpha_{c,i}}\lVert \Tilde{W}_{c,i}^k\lVert^2+\frac{1}{\alpha_{a,i}r_{a,i}}\lVert\Tilde{W}_{a,i}^k\lVert^2 +\frac{1}{\alpha_{d,i}r_{d,i}}\lVert\Tilde{W}_{d,i}^k\lVert^2,
\end{align}
where $r_{a,i} > 0$ and $r_{d,i} > 0$ are the weighting factors.


Consider the critic NN part first by \eqref{e4.2.12} and notice that $\lVert \Tilde{W}_{c,i}^{k+1} \rVert^2 = \lVert \Tilde{W}_{vc,i}^{k+1} \rVert^2$.
One has
\begin{align}\label{e4.3.4}
     \lVert \Tilde{W}_{c,i}^{k+1} \rVert^2 = &\ \lVert \Tilde{W}_{vc,i}^{k}- \alpha_{c,i}e_{c,i,k}\eta(z_{i,k})\rVert^2 \nonumber \\
    =&\ \alpha_{c,i}^2e_{c,i,k}^2\lVert \eta(z_{i,k})\rVert^2 + \lVert\Tilde{W}_{vc,i}^{k}\rVert^2 -2\alpha_{c,i}e_{c,i,k}\Tilde{W}_{vc,i}^{k\top}\eta(z_{i,k})
\end{align}
Note that the third term of \eqref{e4.3.4} can be rewritten using the TD error \eqref{e4.2.3} as
  \begin{align}\label{e4.3.5}
    2e_{c,i,k}\Tilde{W}_{vc,i}^{k\top}\eta(z_{i,k}) 
    = \lVert e_{c,i,k}\rVert^2 +\lVert\Tilde{W}_{vc,i}^{k\top}\eta(z_{i,k})\rVert^2 - \left\lVert r_{i,k} +W_{vc,i}^{*\top}\eta(z_{i,k})\right\rVert^2.
  \end{align}
Taking the first difference of $L_{c,i}^k$ by combining \eqref{e4.3.4} and \eqref{e4.3.5} yields
    \begin{align}\label{e4.3.6}
        \Delta L_{c,i}^k =&\ \frac{1}{\alpha_{c,i}}\left[\lVert\Tilde{W}_{c,i}^{k+1} \lVert^2-\lVert\Tilde{W}_{c,i}^k\lVert^2\right] \nonumber   \\ 
        =&\ \left\lVert r_{i,k} +W_{vc,i}^{*\top}\eta(z_{i,k})\right\rVert^2 -
        \lVert \Tilde{W}_{vc,i}^{k\top}\eta(z_{i,k})\rVert^2 
         -(1-\alpha_{c,i} \lVert \eta(z_{i,k}) \rVert^2) e_{c,i,k}^2 \nonumber  \\ 
        \leq&\  2\lVert r_{i,k} \rVert^2+2\lVert W_{vc,i}^{*\top}\eta(z_{i,k})\rVert^2 -
        \lVert \Tilde{W}_{vc,i}^{k\top}\eta(z_{i,k})\rVert^2 -(1-\alpha_{c,i} \lVert \eta(z_{i,k}) \rVert^2) e_{c,i,k}^2.
    \end{align}
Next, consider the actor NN part.
By similar manipulations to \eqref{e4.3.4}, one has
    \begin{align}\label{e4.3.8}
    \lVert \Tilde{W}_{a,i}^{k+1} \rVert^2 =  &\  \lVert \Tilde{W}_{va,i}^{k}- \alpha_{a,i} \phi_v(\delta_{i,k})  e_{a,i,k}\rVert^2 \notag\\
    = &\ \lVert \Tilde{W}_{va,i}^{k} \rVert^2  - \alpha_{a,i} \left( \lVert e_{a,i,k} \rVert^2 +\lVert \phi_v(\delta_{i,k})^\top \Tilde{W}_{va,i}^{k} \rVert^2  - \lVert \phi_v(\delta_{i,k})^\top W_{va,i}^{*} - u_{i,k}^Q \rVert^2  - \alpha_{a,i}\lVert \phi_v(\delta_{i,k})  e_{a,i,k}\rVert^2 \right). 
    \end{align}

Using \eqref{e4.2.6} and \eqref{e4.3.8}, the first difference of$L_{a,i}^k$ is 
    \begin{align}\label{e4.3.11}
    r_{a,i}\Delta L_{a,i}^k  = &\ \frac{1}{\alpha_{a,i}}\left[\lVert \Tilde{W}_{a,i}^{k+1} \rVert^2-\lVert \Tilde{W}_{a,i}^{k} \rVert^2\right] \notag\\ 
     = &\  -\lVert \Tilde{W}_{va,i}^{k\top} \phi_v(\delta_{i,k})\rVert^2    +
    \alpha_{a,i}\lVert \phi_v(\delta_{i,k})  e_{a,i,k}\rVert^2  +  \lVert W_{va,i}^{*\top} - u_{i,k}^Q \phi_v(\delta_{i,k}) \rVert^2 - \lVert e_{a,i,k} \rVert^2 \notag\\ 
     \leq &\ \lVert W_{va,i}^{*} \rVert^2 \lVert \phi_v(\delta_{i,k}) \rVert^2 + 2 \lVert u_{i,k}^{Q}\rVert \lVert \Tilde{W}_{va,i}^{k\top} \phi_v(\delta_{i,k})\rVert - \lVert \Tilde{W}_{va,i}^{k\top} \phi_v(\delta_{i,k}) \rVert^2  +
     \alpha_{a,i} \lVert u_{i,k}^{Q} \rVert^2 \lVert \phi_v(\delta_{i,k})\rVert^2 \notag \\ 
     &\ - (1-\alpha_{a,i} \lVert \phi_v(\delta_{i,k}) \rVert^2) \lVert W_{va,i}^{k\top} \phi_v(\delta_{i,k}) \rVert ^2- 2\alpha_{a,i} \lVert \phi_v(\delta_{i,k}) \rVert^2  W_{va,i}^{k\top} \phi_v(\delta_{i,k}) u_{i,k}^{Q}.
    \end{align}
    
By Cauchy-Schwarz inequality, one has
    \begin{align}\label{e4.3.12}
  - 2\alpha_{a,i} \lVert \phi_v(\delta_{i,k})\rVert^2  W_{va,i}^{k\top} \phi_v(\delta_{i,k}) u_{i,k}^{Q} 
    \leq \alpha_{a,i}\lVert \phi_v(\delta_{i,k})\rVert^2\left(\lVert u_{i,k}^{Q}\rVert^2+\rVert W_{va,i}^{k\top} \phi_v(\delta_{i,k})\lVert^2\right).
    \end{align}
    
Combining \eqref{e4.3.11} and \eqref{e4.3.12}, the first difference $L_{a,i}^k$ can be written as
    \begin{align*} 
    \Delta L_{a,i}^k \leq &\ 
    \frac{1}{r_{a,i}}\left[\lVert W_{va,i}^{*} \rVert^2 \lVert \phi_v(\delta_{i,k}) \rVert^2 + 2 \lVert u_{i,k}^{Q}\rVert \lVert \Tilde{W}_{va,i}^{k\top} \phi_v(\delta_{i,k})\rVert \right.\notag\\
    & \left.- \lVert \Tilde{W}_{va,i}^{k\top} \phi_v(\delta_{i,k}) \rVert^2 +
     2\alpha_{a,i} \lVert u_{i,k}^{Q} \rVert^2 \lVert \phi_v(\delta_{i,k})\rVert^2- (1-2\alpha_{a,i} \lVert \phi_v(\delta_{i,k}) \rVert^2) \lVert W_{va,i}^{k\top} \phi_v(\delta_{i,k}) \rVert ^2\right].
    \end{align*}

Similarly, the first difference of $L_{d,i}^k$ is bounded by
    \begin{align}\label{e4.3.15}
    \Delta L_{d,i}^k \leq &\
    \frac{1}{r_{d,i}}\left[\lVert W_{vd,i}^{*} \rVert^2 \lVert \varphi_v(\delta_{i,k}) \rVert^2 + 2 \lVert w_{i,k}^{Q}\rVert \lVert \Tilde{W}_{vd,i}^{k\top} \varphi_v(\delta_{i,k})\rVert\right. \notag\\
    & \left. - \lVert \Tilde{W}_{vd,i}^{k\top} \varphi_v(\delta_{i,k}) \rVert^2 +
     2\alpha_{d,i} \lVert w_{i,k}^{Q} \rVert^2 \lVert \varphi_v(\delta_{i,k})\rVert^2  - (1-2\alpha_{d,i} \lVert \varphi_v(\delta_{i,k}) \rVert^2) \lVert W_{vd,i}^{k\top} \varphi(\delta_{i,k}) \rVert ^2\right].
    \end{align}

Substituting \eqref{e4.3.6}, \eqref{e4.3.8} and \eqref{e4.3.15} into the first difference of $L_i^k$ yields
    \begin{align}\label{e4.3.16}
    \Delta L_i^k    = &\ \Delta L_{c,i}^k + \Delta L_{a,i}^k + \Delta L_{d,i}^k \notag \\
    \leq &-(1-\alpha_{c,i} \lVert \eta(z_{i,k}) \rVert^2) e_{c,i,k}^2 - \frac{1}{r_{a,i}} (1-2\alpha_{a,i} \lVert \phi_v(\delta_{i,k}) \rVert^2) \lVert W_{va,i}^{k\top} \phi(\delta_{i,k}) \rVert ^2 - \frac{1}{r_{d,i}} (1-2\alpha_{d,i} \lVert \varphi_v(\delta_{i,k}) \rVert^2) \lVert W_{vd,i}^{k\top} \varphi(\delta_{i,k}) \rVert ^2 \notag \\ 
    &- \lVert \Tilde{W}_{vc,i}^{k\top}\eta(z_{i,k})\rVert^2-
    \frac{1}{r_{a,i}}\lVert \Tilde{W}_{va,i}^{k\top}\phi_v(\delta_{i,k})\rVert^2 -
    \frac{1}{r_{d,i}}\lVert \Tilde{W}_{vd,i}^{k\top}\varphi_v(\delta_{i,k})\rVert^2
   + 2\lVert r_{i,k} \rVert^2+2\lVert W_{vc,i}^{*\top}\eta(z_{i,k})\rVert^2 + \frac{1}{r_{a,i}}\lVert W_{va,i}^{*\top} \rVert^2 \lVert \phi_v(\delta_{i,k}) \rVert^2\notag\\ 
    & + \frac{2}{r_{a,i}} \lVert u_{i,k}^{Q}\rVert \lVert \Tilde{W}_{va,i}^{k\top} \phi_v(\delta_{i,k})\rVert+ \frac{2\alpha_{a,i}}{r_{a,i}} \lVert u_{i,k}^{Q} \rVert^2 \lVert \phi_v(\delta_{i,k})\rVert^2 + \frac{1}{r_{d,i}}\lVert W_{vd,i}^{*\top} \rVert^2 \lVert \varphi_v(\delta_{i,k}) \rVert^2 \notag\\
    & + \frac{2}{r_{d,i}} \lVert w_{i,k}^{Q}\rVert \lVert \Tilde{W}_{vd,i}^{k\top} \varphi_v(\delta_{i,k})\rVert  + \frac{2\alpha_{d,i}}{r_{d,i}} \lVert w_{i,k}^{Q} \rVert^2 \lVert \varphi_v(\delta_{i,k})\rVert^2.
    \end{align}

The first difference of $L_i^k$ can be made negative by selecting $a_{a,i}$, $a_{d,i}$, $a_{c,i}$ such that the first three terms are negative. This is always possible since $\eta(z_{i,k})$, $\phi_v(\delta_{i,k})$ and $\varphi_v(\delta_{i,k})$ are bounded.

Define $$ Z_{i,k} =\col\left(\Tilde{W}_{vc,i}^{k\top}\eta(z_{i,k}),\Tilde{W}_{va,i}^{k\top}\phi_v(\delta_{i,k}),\Tilde{W}_{vd,i}^{k\top}\varphi_v(\delta_{i,k})\right).$$
Then \eqref{e4.3.16} becomes
    \begin{align}\label{e4.3.19}
    \Delta L_{i,k} \leq &\  \bar{\sigma}_2 \lVert Z_{i,k} \rVert- Z_{i,k}^\top \bar{\sigma}_1 Z_{i,k}
     + 2\lVert r_{i,k} \rVert^2 +2\lVert W_{vc,i}^{*\top}\eta(z_{i,k})\rVert^2 + \frac{2\alpha_{a,i}}{r_{a,i}} \lVert u_{i,k}^{Q} \rVert^2 \lVert \phi_v(\delta_{i,k})\rVert^2 \notag \\
    & + \frac{2\alpha_{d,i}}{r_{d,i}} \lVert w_{i,k}^{Q} \rVert^2 \lVert \varphi_v(\delta_{i,k})\rVert^2 + \frac{1}{r_{a,i}}\lVert W_{va,i}^{*\top} \rVert^2 \lVert \phi_v(\delta_{i,k}) \rVert^2+ \frac{1}{r_{d,i}}\lVert W_{vd,i}^{*\top} \rVert^2 \lVert \varphi_v(\delta_{i,k}) \rVert^2
    \end{align}
where 
    \begin{align*}
\bar{\sigma}_1 =\min\left\{1, \frac{1}{r_{a,i}}, \frac{1}{r_{d,i}}\right\},\;\;\;\;
\bar{\sigma}_2 =\max\left\{ \frac{2 \lVert \bar{u}_i\rVert}{r_{a,i}}, \frac{2 \lVert \bar{w}_i\rVert}{r_{d,i}} \right\}.
    \end{align*}

Let $\bar{r}_i$, $\bar{\eta}_i$, $\bar{\phi}_i$ and $\bar{\varphi}_i$be the upper bound of $r_i$, $\eta(z_{i,k})$, $\phi_v(\delta_{i,k})$, $\varphi_v(\delta_{i,k})$.
Considering the upper bound of each term except $Z_{i,k}$ in \eqref{e4.3.19}, the following bound is obtained
    \begin{align}\label{e4.3.20}
   2\lVert& r_{i,k} \rVert^2+2\lVert W_{vc,i}^{*\top}\eta(z_{i,k})\rVert^2  
   + \frac{1}{r_{a,i}}\lVert W_{va,i}^{*\top} \rVert^2 \lVert \phi_v(\delta_{i,k}) \rVert^2  + \frac{2\alpha_{a,i}}{r_{a,i}} \lVert u_{i,k}^{Q} \rVert^2 \lVert \phi_v(\delta_{i,k})\rVert^2 \notag\\ 
    & +  \frac{1}{r_{d,i}}\lVert W_{vd,i}^{*\top} \rVert^2 \lVert \varphi_v(\delta_{i,k}) \rVert^2 + \frac{2\alpha_{d,i}}{r_{d,i}} \lVert w_{i,k}^{Q} \rVert^2 \lVert \varphi_v(\delta_{i,k})\rVert^2 \notag\\ 
    \leq &\
    2 \bar{r}_i^2 + 2\bar{\eta}_i^2 \bar{W}_{c,i}^2 + \frac{1}{r_{a,i}} \bar{W}_{a,i}^2 \bar{\phi}_i^2 + \frac{2\alpha_{ai}}{r_{a,i}} \bar{u}_{i}^2 \bar{\phi}_i^2 + \frac{1}{r_{d,i}} \bar{W}_{d,i}^2 \bar{\varphi}_i^2 + \frac{2\alpha_{d,i}}{r_{d,i}} \bar{w}_{i}^2 \bar{\varphi}_i^2 
    =  M.
    \end{align}
Combining \eqref{e4.3.19} and \eqref{e4.3.20}, we have:
    \begin{align}
    \Delta L_{i}^k  \leq
    - Z_{i,k}^\top \bar{\sigma}_1 Z_{i,k}
     + \bar{\sigma}_2 \lVert Z_{i,k} \rVert + M.\notag
    \end{align}
Therefore, $\Delta L_{i}^k<0$ if 
\begin{align*}
    \lVert Z_{i,k} \rVert  > \frac{\bar{\sigma_2}}{2\bar{\sigma_1}} + \sqrt{\frac{M}{\bar{\sigma_1}} + \frac{\bar{\sigma_2}^2}{4\bar{\sigma_1}^2}} \equiv B_z
\end{align*}
and the learning rates satisfy $a_{c,i} \geq 1/\bar{\eta}_i^2$, $a_{a,i} \geq 1 /2\bar{\phi}_i^2$ and $a_{d,i} \geq 1 /2\bar{\varphi}_i^2$.
Combing $\eqref{e4.2.3}$ and $\eqref{e4.2.12}$ leads to a virtual system
\begin{align} \label{e4.3.21}
   \Tilde{W}_{vc,i}^{k+1} & = A_{c,i,k}\Tilde{W}_{vc,i}^{k} + B_{c,i,k}u_{c,i,k} \notag \\
    y_{c,i,k} & = C_{c,i,k} \Tilde{W}_{vc,i}^{k}
\end{align}
where $A_{c,i,k} = (I - \alpha_{c,i} \eta(z_{i,k}) \eta(z_{i,k})^\top)$, $B_{c,i,k} = -\alpha_{c,i}\eta(z_{i,k})$, $C_{c,i,k}=\eta(z_{i,k})^\top$ and $u_{c,i,k} = r_{i,k} + W_{vc,i}^{*\top} \eta(z_{i,k})$.
The PE condition for $\eta(z_{i,k})$ is equivalent to the uniformly completely observability (UCO) of $(I, C_{c,i,k})$. 
Moreover, the output feedback does not change the UCO property of $(A_{c,i,k}, C_{c,i,k})$.
If we use the output feedback, the system \eqref{e4.3.21} becomes 
\begin{align*}
    \Tilde{W}_{vc,i}^{k+1} & = A_{c,i,k}\Tilde{W}_{vc,i}^{k} + B_{c,i,k}(u_{c,i,k} - y_{c,i,k})\\
    & = \Tilde{W}_{vc,i}^{k} + B_{c,i,k}u_{c,i,k} \\
    y_{c,i,k} & = C_{c,i,k} \Tilde{W}_{vc,i}^{k}.
\end{align*}
The PE condition guarantees that $\Tilde{W}_{vc,i}^{k}$ is bounded if $y_{c,i,k}$ is bounded.
The similar analysis can be taken to $\Tilde{W}_{va,i}^{k}$ and $\Tilde{W}_{vd,i}^{k}$. 
Since $Z_{i,k}$ is UUB, the critic NN errors $\Tilde{W}_{c,i}^k$, the actor NN errors $\Tilde{W}_{a,i}^k$ and the disturber NN errors $\Tilde{W}_{d,i}^k$ are UUB.
It can be concluded that $\hat{Q}_{i}$ converges to the approximate optimal value of the cooperative zero-sum graphical game and $\hat{u}_i$, $\hat{w}_i$ converge to the approximate zero-sum Nash equilibrium solution. This completes the proof.\hfill$\Box$

\section*{Acknowledgments}
This work was supported by the National Natural Science Foundation of China under grant 62473114, the Guangdong Basic and Applied Basic Research Foundation under project 2023A1515011981, the Shenzhen Science and Technology Program under projects JCYJ20220818102416036 and RCJC20210609104400005, and partly by NSERC.

\section*{Conflict of Interest}
The authors declare that there is no conflict of interest on the article.

\section*{Data Availability Statement}
Data sharing not applicable to this article as no datasets were generated or analyzed during the current study.

\end{document}